\documentclass[a4paper,USenglish,cleveref, autoref]{lipics-v2019}

\title{Approximate Learning of Limit-Average Automata}

\author{Jakub Michaliszyn}{University of Wroc\l{}aw}{jmi@cs.uni.wroc.pl}{https://orcid.org/0000-0002-5053-0347}{} 
\author{Jan Otop}{University of Wroc\l{}aw}{jotop@cs.uni.wroc.pl}{https://orcid.org/0000-0002-8804-8011}{}

\authorrunning{Jakub Michaliszyn and Jan Otop}

\Copyright{Jakub Michaliszyn and Jan Otop}

\ccsdesc[500]{Theory of computation~Automata over infinite objects}
\ccsdesc[500]{Theory of computation~Quantitative automata} 

\keywords{weighted automata, learning, expected value}

\nolinenumbers 

\hideLIPIcs

\funding{The National Science Centre (NCN), Poland under grant 2017/27/B/ST6/00299.}

\usepackage{amssymb}
\usepackage[colorinlistoftodos,prependcaption,disable]{todonotes}
\presetkeys{todonotes}{inline}{}
\usepackage{tikz}
\usepackage{environ}
\usetikzlibrary{arrows,automata,shapes,positioning}
\usetikzlibrary{positioning,decorations.markings}

\usepackage{lineno}

\newtheorem{problem}[theorem]{Problem}

\usepackage{thmtools,thm-restate}
\newcommand{\set}[1]{\{#1\}}

\newcommand{\aut}{\mathcal{A}}
\newcommand{\Q}{\mathbb{Q}}

\newcommand{\N}{\mathbb{N}}
\newcommand{\R}{\mathbb{R}}
\newcommand{\flimavg}{\textsc{LimAvg}}
\newcommand{\fsum}{\textsc{Sum}}

\newcommand{\favg}[1]{\textsc{Avg}(#1)}

\newcommand{\Paragraph}[1]{\noindent\textbf{#1.}}

\newcommand{\NP}{\textsc{NP}{}}

\newcommand{\prob}{\mathbb{P}}
\newcommand{\expected}{\mathbb{E}}
\newcommand{\valueL}[1]{\mathcal{L}_{{#1}}}

\newcommand{\cost}{{C}}
\newcommand{\markov}{\mathcal{M}}

\newcommand{\tuple}[1]{\langle #1 \rangle}

\newcommand{\lang}{\mathcal{L}}

\newcommand{\uniformN}[1]{\mathcal{U}^{#1}}
\newcommand{\uniformFin}[1]{\mathcal{G}(#1)}
\newcommand{\uniformInf}{\mathcal{U}^{\infty}}

\tikzset{
    position/.style args={#1:#2 from #3}{
        at=(#3.#1), anchor=#1+180, shift=(#1:#2)
    }
}

\newcommand{\satz}{SAT${}_{\mathrm{0}}$}

\newcommand{\lInfNorm}[1]{\left\lVert#1\right\rVert_{\infty}}

\usepackage{etoolbox,refcount}
\usepackage{multicol}

\newcounter{countitems}
\newcounter{nextitemizecount}
\newcommand{\setupcountitems}{
  \stepcounter{nextitemizecount}
  \setcounter{countitems}{0}
  \preto\item{\stepcounter{countitems}}
}
\makeatletter
\newcommand{\computecountitems}{
  \edef\@currentlabel{\number\c@countitems}
  \label{countitems@\number\numexpr\value{nextitemizecount}-1\relax}
}
\newcommand{\nextitemizecount}{
  \getrefnumber{countitems@\number\c@nextitemizecount}
}
\newcommand{\previtemizecount}{
  \getrefnumber{countitems@\number\numexpr\value{nextitemizecount}-1\relax}
}
\makeatother    
\newenvironment{AutoMultiColItemize}{
\vspace{-1.2em}
\ifnumcomp{\nextitemizecount}{>}{3}{\begin{multicols}{2}}{}
\setupcountitems\begin{itemize}}
{\end{itemize}
\unskip\computecountitems\ifnumcomp{\previtemizecount}{>}{3}{\end{multicols}}{}}
\clearpage{}
\begin{document}

\maketitle

\begin{abstract}
Limit-average automata are weighted automata on infinite words that use average to aggregate the weights seen in infinite runs.  
We study approximate learning problems for limit-average automata in two settings: passive and active. 
In the passive learning case, we show that limit-average automata are not PAC-learnable as samples must be of exponential-size to provide (with good probability) enough details to learn an automaton. 
We also show that the problem of finding an automaton that fits a given sample is \NP{}-complete.
In the active learning case, we show that limit-average automata can be learned almost-exactly, i.e., we can learn in polynomial time an automaton that is consistent with the target automaton on almost all words. 
On the other hand, we show that the problem of learning an automaton that approximates the target automaton (with perhaps fewer states) is \NP{}-complete.
The abovementioned results are shown for the uniform distribution on words. We briefly discuss learning over different distributions.
 \end{abstract}

\newcommand{\introPara}[1]{}
\section{Introduction}

 Quantitative verification has been proposed to verify non-Boolean system properties such as performance, energy consumption, fault-tolerance, etc.
There are two main challenges in applying quantitative verification in practice. 
First, formalization of quantitative properties is difficult as specifications are given directly as weighted automata, which extend finite automata with weights on transitions~\cite{Droste:2009:HWA:1667106,quantitativelanguages}. Quantitative logics, which could facilitate the specification task,
have either limited expression power or their model-checking problem is undecidable~\cite{BokerLTLCTL,avgLTL}. 
Second, there is little research on abstraction in the quantitative setting~\cite{CernyHR13}, which would allow us to reduce the size of quantitative models, presented as weighted automata.

We approach both problems using learning of weighted automata. We can apply the learning framework to facilitate writing quantitative specifications and make it more accessible to non-experts. 
For abstraction purposes, we study approximate learning, having in mind that approximation of a system can significantly reduce its size.

We focus on weighted automata over infinite words which compute the long-run average of weights~\cite{quantitativelanguages}.
Such automata express interesting quantitative properties~\cite{quantitativelanguages,henzingerotop17} and admit polynomial-time algorithms for basic decision questions~\cite{quantitativelanguages}.
Some of the interesting properties can be only expressed by non-deterministic automata~\cite{quantitativelanguages}, but every non-deterministic weighted automaton is approximated by some deterministic weighted automaton on almost all words~\cite{MichaliszynOtop18}.
This means that, by allowing some small margin of error, we can focus on deterministic automata 
while being able to model important system properties, such as minimal response time, minimal number of errors, the edit distance problem~\cite{henzingerotop17}, and the \emph{specification repair} framework from~\cite{DBLP:conf/lics/BenediktPR11}.

We therefore focus on approximate learning under probabilistic semantics, which corresponds to the average-case analysis in quantitative verification~\cite{ChatterjeeHO16}. 
We treat words as random events and functions defined by weighted automata as random variables.
In this setting, an automaton $\aut'$ $\epsilon$-approximates $\aut$ if the expected difference between $\aut$ and $\aut'$ (over all words) is bounded by $\epsilon$. 
We consider two learning approaches: passive and active.

In passive learning, we think of an automaton as a black-box. 
This might be a program, a specification, a working correct system to be replaced or abstracted, or even only some examples of how the automaton should work. 
Our goal is to construct a weighted automaton based on this black-box model by observing only inputs and outputs of the model.

In active learning, we assume presence of an interactive \emph{teacher} that
reveals the values of given examples 
and 
verifies whether a provided automaton is as intended; if not, the teacher responses with a witness, which is a word showing the difference between the constructed automaton and the intended automaton.
This does not necessary mean that the teacher is familiar with the weighted automata formalism --- to provide a witness, they may simply run the constructed automaton in parallel with the black-box and, 
if at any point their behaviors differ, provide the appropriate input to the learning algorithm.

\Paragraph{Our contributions}
We start with a discussion on the setting for learning problems. 
The first step is to find a suitable representation for samples; not all infinite words have finite representations.
A natural idea is to use ultimately periodic words in samples; we study such samples. 
However, the probability distribution over ultimately periodic words is very different from the uniform distribution over infinite words.
Therefore, we also consider samples consisting of a finite word $u$ labeled with the expected value over all extensions of $u$.
The probability distribution over such samples is closer to the uniform distribution over infinite words.

Then we study the passive learning problem. We show that for unique characterization of an automaton, we need a sample of exponential size. 
We study the complexity of the problem of finding an automaton that fits the whole sample. 
The problem, without additional restrictions, has a trivial and overfitting solution. 
To mitigate this, we impose bounds on automata size, and show that then the problem is \NP{}-complete.

For active learning, we show that the problem of learning an almost-exact automaton can be solved in polynomial time, and finding an automaton of bounded size that approximates the target one cannot be done in polynomial time if P $\neq$ \NP{}. We conclude with a discussion on different probability distributions.

\Paragraph{Related work}
The \emph{probably approximately correct} (PAC) learning, introduced in~\cite{valiant1984theory}, is a general passive learning framework applied to various objects (DNF/CNF formulas, decision trees, automata, etc.)~\cite{kearns1994introduction}.
PAC learning of deterministic finite automata (DFA) has been extensively studied despite negative indicators. First, the \emph{sample fitting} problem for DFA, where the task is to construct a minimal-size 
DFA consistent with a given sample, has been shown \NP{}-complete~\cite{gold1978complexity}. Even approximate sample fitting, where we ask for a DFA at most polynomially greater than a miniaml-size DFA, remains \NP-complete~\cite{pitt1993minimum}.
Second, it has been shown that existence of a polynomial-time PAC learning algorithm for DFA would break certain cryptographic systems (such as RSA) and hence it is unlikely~\cite{kearns1994cryptographic}.
Despite these negative results, it has been empirically shown that DFA can be efficiently learned~\cite{Lang92}. In particular, if we assume \emph{structural completeness} of a sample, then 
it determines a minimal DFA~\cite{oncina1992inferring}.
Pitt posed a question whether DFA are PAC-learnable under the uniform distribution~\cite{pitt1989inductive}, which remains open~\cite{Angluin15}.

Angluin showed that DFA can be learned in polynomial time, if the learning algorithm can ask membership and equivalence queries~\cite{angluin1987learning}.
This approach proved to be very fruitful and versatile. Angluin's algorithm has been adapted to 
learn NFA~\cite{BolligHKL09}, automata over infinite words~\cite{AngluinF16,Fisman18}, nominal automata over infinite alphabets~\cite{MoermanS0KS17}, weighted automata over words~\cite{activeWeightedAutomata} and weighted automata over trees~\cite{HabrardO06,MarusicW15}.

Recently, there has been a renewed interest in learning weighted automata~\cite{BallePP15,MarusicW15,BalleM15a,BalleM15b,BalleM18}. 
These results apply to weighted automata over fields~\cite{Droste:2009:HWA:1667106}, which work over finite words. 
We, however, consider limit-average automata, which work over infinite words and cannot be defined using a field or even a semiring.
Furthermore, we consider weighted (limit-average) automata under probabilistic semantics~\cite{ChatterjeeHO16,MichaliszynOtop18}, i.e., we consider functions represented by automata 
as random variables.

\section{The setting}

\label{sec:samples}
Given a finite alphabet $\Sigma$ of letters, a \emph{word} $w$ is a finite or infinite sequence 
of letters.
We denote the set of all finite words over $\Sigma$ by $\Sigma^*$, and the set of all infinite words over $\Sigma$ by $\Sigma^\omega$.
For a word $w$, we define $w[i]$ as the $i$-th letter of $w$, and we define $w[i,j]$ as the subword $w[i] w[i+1] \ldots w[j]$ of $w$. 
We use the same notation for vectors and sequences; we assume that sequences start with $0$ index.

A \emph{deterministic finite automaton} (DFA) is a tuple $(\Sigma, Q, q_0, F, \delta)$
consisting of
	the alphabet $\Sigma$, 
	a finite set of states $Q$, 
	the initial state $q_0 \in Q$,  
	a set of final states $F$, and
    a transition function 	$\delta \colon Q \times \Sigma \to Q$.

A (deterministic) $\flimavg$-automaton extends a DFA 
with a function $\cost \colon \delta \to \mathbb{Q}$  that defines rational \emph{weights} of transitions. 
The size of such an automaton $\aut$, denoted by $|\aut|$,
 is the sum of the number of states of $\aut$ and the lengths of binary encodings of all the weights.

A \emph{run} $\pi$  of a $\flimavg$-automaton $\aut$ on a word $w$ is a sequence of states $\pi[0] \pi[1] \dots$ such that $\pi[0]$ is the initial state and for every $i>0$ we have $\delta(\pi[i-1],w[i]) = \pi[i]$.
We do not consider $\omega$-accepting conditions and assume that all infinite runs are accepting.
Every run $\pi$ of $\aut$ on an infinite word $w$ defines a sequence of weights $\cost(\pi)$
of successive transitions of $\aut$, 
i.e., $\cost(\pi)[i] = \cost(\pi[i-1], w[i], \pi[i])$. 
The \emph{value} of the run $\pi$ is then defined as 
\(\flimavg(\pi) = \limsup_{k \rightarrow \infty} \favg{\cost(\pi)[0, k]} \),
where for finite runs $\pi$ we have \(\favg{\cost(\pi)}={\fsum(\cost(\pi))}/{|\cost(\pi)|}\).
The value of a word $w$ assigned by the automaton $\aut$, denoted by $\valueL{\aut}(w)$,
is the value of the run of $\aut$ on $w$.

We consider three classes of probability measures on words over the alphabet $\Sigma$. 
\begin{itemize}
\item $\uniformN{n}$, for $n \in \N$, is the uniform probability distribution on $\Sigma^n$ assigning each word the probability $|\Sigma|^{-n}$.
\item $\uniformFin{\lambda}$, for a termination probability $\lambda \in \Q^+ \cap (0,1)$, is such that for $u \in \Sigma^*$ we have
$\uniformFin{\lambda}(u) = |\Sigma|^{-|u|} \cdot (1-\lambda)^{|u|} \cdot \lambda$. Observe that the probability of words of the same length is equal, and the probability of
generating a word of length $k$ is $(1-\lambda)^k \cdot \lambda$. We can consider this as process generating finite words, which stops after each step with probability $\lambda$.
\item $\uniformInf$ is the uniform probability measure on $\Sigma^{\omega}$. Formally, we define $\uniformInf$ on basic sets $u \Sigma^{\omega} = \{ u w \mid w \in \Sigma^{\omega} \}$ as follows:
$\uniformInf(u \Sigma^{\omega}) = |\Sigma|^{-|u|}$. Then, $\uniformInf$ is the unique extension on all Borel sets in $\Sigma^{\omega}$ considered with the product topology~\cite{feller,BaierBook}.
\end{itemize}

\Paragraph{Automata as random variables}
A weighted automaton $\aut$ defines the measurable function $\valueL{\aut}(w) \colon \Sigma^\omega \mapsto \R$ that assigns values to words.
We interpret such functions as random variables w.r.t. the probabilistic measure $\uniformInf$.
Hence, for a given automaton $\aut$, we consider the following quantities:
\begin{enumerate}
\item $\expected(\aut)$ --- the expected value of
the random variable $\valueL{\aut}$ defined by $\aut$ w.r.t. the uniform distribution $\uniformInf$ on $\Sigma^{\omega}$, and
\item $\expected(\aut \mid u\Sigma^{\omega})$ --- the conditional expected value, defined for $\uniformInf$ as the expected value 
of $\lang^u$ such that $\lang^u(w) = \valueL{\aut}(uw)$.
\end{enumerate}

We consider automata as generators of random variables --- two automata are \emph{almost equivalent} if they define the almost equal random variables. 
Formally, we say that $\aut_1$ and $\aut_2$ are \emph{almost equivalent} if and only if 
for almost all words $w$ we have $\valueL{\aut_1}(w) = \valueL{\aut_2}(w)$.
Note that almost all words means all except for words from some $Y$ of probability $0$.

$\flimavg$-automata considered over probability distributions are equivalent to Markov chains with the long-run average objectives presented in~\cite[Section~10.5.2]{BaierBook}.

\begin{theorem}[\cite{BaierBook}]
Let $\aut$ be a $\flimavg$-automaton. 
(1)~If $\aut$ is strongly connected, then almost all words have the same value equal to $\expected(\aut)$.
(2)~For almost all words $w$, the run on $w$ eventually reaches some bottom strongly connected component (SCC) of $\aut$.
\label{f:ExpectedLIMAVG}
\end{theorem}

Theorem~\ref{f:ExpectedLIMAVG} has important consequences. 
First, we can contract each bottom SCC to a single state with self-loops of the same weight $x$ being the expected value of that SCC. 
We refer to $x$ as the value of that SCC.
Such an operation does not affect almost equivalence. 
Second, while reasoning about $\flimavg$-automata, we can neglect all weights except for the values of bottom SCCs.
We will omit weights other than the values of bottom SCCs.

\newcommand{\autTarget}{\aut^T}

\Paragraph{Samples}
A sample is a set of \emph{labeled examples} of some (hidden) function $\lang \colon \Sigma^{\omega} \to \R$. 
In the classical automata-learning approach words are finite, and hence they can be presented as examples along with the information  whether the word belongs to the hidden language or not.
In the infinite-word case, however, words cannot be given directly to a learning algorithm. We present two alternative solutions for this problem. One is to restrict examples to ultimately periodic words, which have finite presentation, and the other is to consider finite words and ask for conditional expected values. We discuss both approaches below. 

To distinguish samples with different types of labeled examples, we call them $U$-samples, $E$-samples and $(E,n)$-samples.
\begin{enumerate}
\item \emph{Ultimately periodic words}. Consider an example being an ultimately periodic word $uv^{\omega}$. 
It can presented as a pair of finite words $(u,v)$ and
we consider labeled examples $(u,v,x)$, where $u,v \in \Sigma^*$ and $x = \lang(uv^{\omega})$.
A set of labeled examples $(u,v,x)$ is called an $U$-sample. 

To draw a random $U$-sample, we consider finite words $u,v$ to be selected independently at random according to distributions $\uniformFin{\lambda_1}$ and $\uniformFin{\lambda_2}$ for some $\lambda_1, \lambda_2$.
For such a set of pairs of words, we label them according to the function $\lang$.  

\item \emph{Conditional expected values}. 
We consider examples, which are finite words $u \in \Sigma^*$. 
A labeled example is a pair $(u,x)$, where $x = \expected(\lang \mid u\Sigma^{\omega})$ is the conditional expected value of $\lang$ under the condition that random words start with $u$.
For such labeled examples we consider $E$-samples consisting of labeled words of various length, and  
$(E,n)$-samples consisting of words of length $n$.
We assume that finite words for random $E$-samples are drawn according to a distribution $\uniformFin{\lambda}$ for some $\lambda$, and
finite words for random $(E,n)$-samples are drawn according to the uniform distribution $\uniformN{n}$.
\end{enumerate}

We only consider minimal consistent samples, i.e., samples that do not contain examples whose value can be computed from other examples in the sample. 
For instance, $\set{(a, a, 0), (a, aa, 1)}$ is an inconsistent $U$-sample, and  $\set{(aa, 0), (ab, \frac{1}{2}), (a, 1)}$ is an inconsistent $E$-sample over $\set{a, b}$.

\begin{remark}[Incompatible distributions]
Note that the distribution on ultimately periodic words differ from the uniform distribution on infinite words. 
The set of ultimately periodic words is countable and hence it has probability $0$ (according to the distribution $\uniformInf$).
Moreover, almost all infinite words contain all finite words as infixes, whereas this is not the case for ultimately periodic words under any probability distribution.
\end{remark}

\begin{remark}[Feasibility of conditional expectation]
Consider a $\flimavg$-automaton $\aut$ computing $\lang \colon \Sigma^\omega \to \R$. 
For a finite prefix $u$, we can compute $\expected(\lang \mid u\Sigma^{\omega})$ in polynomial time in $|\aut|$~\cite{BaierBook}.
If we consider $\aut$ to be a black-box, which can be controlled, then $\expected(\lang \mid u\Sigma^{\omega})$ can be approximated in the following way.
We pick random words $v_1, \ldots, v_k$ of length $k$, compute partial averages in $\aut$ of $uv_1, \ldots, uv_k$ and then take the average of these values. 
The probability that this process returns a value $\epsilon$-close to $\expected(\lang \mid u\Sigma^{\omega})$ converges to $1$ at exponential rate with $k$.
\end{remark}
 
\section{Passive learning}
\label{sec:passive}
\newcommand{\autInd}{\mathcal{A}}
\newcommand{\autIndc}{\mathcal{\overline{A}}}
\newcommand{\len}[1]{||#1||}
\newcommand{\Sref}[1]{\textbf{\textsc{S\ref*{#1}}}}

Passive learning corresponds to a scenario with an uncontrolled  working black-box system. The learner can only observe system's output, and its goal is to create an approximate model of the system. 
This task comprise of two problems. The first problem, \emph{characterization}, is to assess whether the observations cover most, if not all, behaviors of the system. 
The second one, called \emph{sample fitting}, is to create a reasonable automaton consistent with the observations.
In this section we discuss both problems.

\subsection{Characterization}
A sample can cover only small part of the system. 
It is sometimes argued~\cite{Lang92,oncina1992inferring}, however, that if a sample is large enough, then it is likely to cover most, if not all,
 important behaviors.
 We show that for some $\flimavg$-automata, 
 randomly drawn samples of size less than exponential are unlikely to demonstrate any probable values.

 Let $\len{S}$ denote the sum of the lengths of all the examples in a sample $S$. 
 A sample \emph{distinguishes} two automata if it is consistent with exactly one of them.  We show the following.

\begin{theorem}\label{thm:characterisation}
For any $n$ there are two automata $\autInd_n$, $\autIndc_n$ of size $n+4$ such that for almost all words $w$ 
we have $|\valueL{\autInd_n}(w)-\valueL{\autIndc_n}(w)|=2$, but 
a random $U$-sample, $E$-sample or $(E, k)$-sample (for any $k$) $S$ distinguishes $\autInd_n$ and $\autIndc_n$ with the probability at most $\frac{\len{S}}{2^n}$.
\end{theorem}

\begin{proof}
Consider the alphabet $\set{a, b}$ and $n \in \N$. 
We construct a $\flimavg$-automaton $\autInd_n$ with $n+4$ states. We use $n+2$ states to find the first occurrence of the infix $a^nb$ in the standard manner. When such an infix is found, the automaton moves to a state $q_a$ if the following letter is $a$ and to a state $q_b$ otherwise. In $q_a$ it loops with the weight $1$ and in $q_b$ it loops with the weight $-1$. All other weights are $0$. 
So $\autInd_n$ returns $1$ if the first occurrence of $a^nb$ if followed by $a$, $-1$ if it is followed by $b$, and $0$ if there is no $a^nb$.
The automaton $\autIndc_n$ has the same structure as $\autInd_n$, but the weights $-1$ and $1$ are swapped.

We first observe that $\autInd_n$ and $\autIndc_n$ differ over almost all infinite words.  
Indeed, an infinite word with probability $1$ contains the infix $a^{n}b$, and so on almost all words
one of the automata returns $-1$ 
and the other one $1$.
Consider a sample $S$ (it can be an $E$-sample, $(E,k)$-sample or $U$-sample).
This sample distinguishes automata $\autInd_n$ and $\autIndc_n$ only if it contains an example with the infix $a^{n}b$ (in case of $U$-samples, this means that this infix occurs in $uv$ of some example $(u,v)$) - all other examples for both automata are the same. 
The probability that $S$ contains $a^nb$ as an infix of one of its examples is bounded by $\frac{\len{S}}{2^{n+1}}$. Indeed, the number of positions in all words is $\len{S}$ and the probability that $a^nb$ occurs on some specific position is at most $\frac{1}{2^{n+1}}$
for all types of samples and any $k >0$ (if $k<n+1$, the probability is $0$). 
\end{proof}

Therefore to be able to distinguish just two automata with probability $1-\epsilon$, 
we need a sample such that $\frac{\len{S}}{2^{n+1}} >1-\epsilon$, which for fixed $\epsilon<1$ is of exponential size.
We show that the exponential upper bound is sufficient.

\Paragraph{$U$-samples} If automata $\aut_1$, $\aut_2$ of size $n$ recognize different languages, 
then there is a word $uw^\omega$ such that $\valueL{\aut_1}(uw^\omega) \neq \valueL{\aut_2}(uw^\omega)$ and 
the length of $u$ and $w$ is bounded by $n^2$. 
Assume for simplicity that $\lambda = \lambda_1 = \lambda_2$.
If the sample size is at least 
$|\Sigma|^{2n^2}  \cdot \ln \frac{|\Sigma|^{2n^2}}{\epsilon}     \cdot    (1-\lambda)^{-2n^2}               \cdot \lambda^{-2}$, 
then with probability $1-\epsilon$ a random sample contains all such words, and so distinguishes all 
the automata of size $n$.

\Paragraph{$E$-samples} $E$-samples do not distinguish almost-equivalent automata, hence we cannot learn automata exactly. 
However, exponential samples are enough to learn automata up to almost equivalence. 
To see that 
consider two automata $\aut_1$ and $\aut_2$  of size $n$ that are not almost equivalent.
Due to Theorem~\ref{f:ExpectedLIMAVG} there is a word $u \in \Sigma^*$ such that 
$u$ reaches bottom SCCs in both $\aut_1$ and $\aut_2$, and these bottom SCCs have different expected values.
Using standard pumping argument, we can reduce the size of $u$ to at most $n^2$. 
So if the sample size is at least $|\Sigma|^{n^2}  \cdot \ln \frac{|\Sigma|^{n^2}}{\epsilon}     \cdot    (1-\lambda)^{-n^2}               \cdot \lambda^{-1}$, then with probability $1-\epsilon$ it contains all words $u$ of size $n^2$, and therefore distinguishes all automata of size $n$.

For $(E, n)$-samples, the reasoning is the similar, assuming $n$ is quadratic in the size of the automaton: the sufficient sample size is $|\Sigma|^{n} \cdot \ln \frac{|\Sigma|^{n}}{\epsilon}$.

\subsection{Consequences for PAC learning}
We discuss the consequences of our results to the \emph{probably approximately correct} (PAC) model of learning~\cite{valiant1984theory}.
In the PAC framework, the learning algorithm should work independently of the probability distribution on samples. 
However, variants of the PAC framework have been considered where the distribution on samples is uniform~\cite{efficientDNF}.
In particular, PAC learning of DFA under the uniform distribution over words is a long-standing open problem~\cite{pitt1989inductive,Angluin15}.

We restrict the classical PAC model and assume that observations are drawn according to the distributions $\uniformN{n}, \uniformFin{\lambda}$ (as discussed in Section~\ref{sec:samples}) 
and the quality of the learned automaton is assessed using the uniform distribution over infinite words $\uniformInf$.

\begin{problem}[PAC learning under fixed distributions]
Given $\epsilon, \delta \in \Q^+$, $n \in \N$ and an oracle returning random labeled examples consistent with some automaton $\autTarget$ of size $n$, construct an automaton $\aut$ such that
with probability $1-\epsilon$ w e have $\expected(|\valueL{\aut} - \valueL{\autTarget}|) < 1 - \delta$.
\end{problem}

As a consequence of Theorem~\ref{thm:characterisation}, there is no PAC-learning algorithm for $\flimavg$-automata with $U$-samples (resp., $E$-samples or $(E,k)$-samples) that uses samples of polynomial size; in particular, there is no such algorithm working in polynomial time.

\begin{theorem}
The class of $\flimavg$-automata is not PAC-learnable  with $U$-samples, $E$-samples or $(E,k)$-samples.
\end{theorem}

\subsection{Sample fitting}
Once we have a sample, the problem of finding an automaton fitting the sample can be solved in polynomial time in a trivial way: we create an automaton that is a tree such that every word of a given sample leads to a different leaf in this tree, and then we add loops with appropriate values in the leaves
(similarly to a prefix tree acceptor~\cite{GrammaticalInference} for finite automata). 
This solution leads to an automaton that \emph{overfits} the samples, as it works well only for the sample and it is unlikely to work well with words not included in the sample. 
Besides, the automaton is linear in the size of the sample, not in the size of the black-box system.
For a fixed automaton we can construct arbitrarily large U-samples (or E-samples) consistent with it and hence the gap between the size of such an automaton and the black-box system is arbitrarily large.
To exclude such solutions, we restrict the size of the automaton to be constructed.
 We study the following problem.

\begin{problem}[Sample fitting]
Given a sample $S$ and $n \in \N$, 
construct a $\flimavg$-automaton  
with at most $n$ states, 
which is consistent with $S$.
\end{problem}

The decision version of this problem only asks whether such an automaton exists.
We show that this problem is \NP{}-complete, regardless of the sample representation.
For hardness, we reduce the NP-complete problem \satz{} \cite{GrammaticalInference}, which is the SAT problem restricted to CNF formulas such that each clause contains only positive literals or only negative literals.

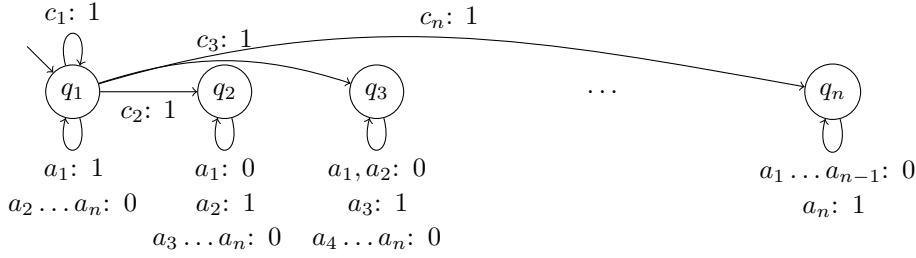
\begin{figure}
\begin{tikzpicture}[state/.style={draw,circle}]
\newdimen\nodeDist
\nodeDist=6mm

\node[state] (c1) at (0,0) {$q_1$};
\node[state] (c2) at (2,0) {$q_2$};
\node[state] (c3) at (4,0) {$q_3$};
\node[state] (cn) at (10,0) {$q_n$};
\node (dots) at (7, 0) {\dots};
\draw[->] (-0.6, 0.6) to (c1);

\draw[->] (c1) to node[below]{$c_2$: 1} (c2);
\draw[->] (c1) to[bend left,out=18,in=162] node[above]{$c_3$: 1} (c3);
\draw[->] (c1) to[bend left,out=17,in=170] node[above]{$c_n$: 1} (cn);
\draw[->,loop above] (c1) to node[above]{$c_1$: 1} (c1);

\draw[->,loop below] (c1) to node[below,text width=1.7cm] {\centering{}$a_1$: 1\\$a_2 \dots a_n$: 0} (A);
\draw[->,loop below] (c2) to node[below,text width=1.9cm] {\centering{}$a_1$: 0\\$a_2$: 1\\\centering{} $a_3 \dots a_n$: 0} (A);
\draw[->,loop below] (c3) to node[below,text width=1.7cm] {\centering{}$a_1, a_2$: 0\\ \centering{} $a_3$: 1  \\ $a_4 \dots a_n$: 0} (A);
\draw[->,loop below] (cn) to node[below,text width=2.2cm] {\centering{} $a_1 \dots a_{n-1}$: 0 \\ \centering{} $a_n$: 1\\} (A);
\end{tikzpicture}
\caption{
The canonical automaton $\aut_{\varphi}$ from the proof of Theorem~\ref{thm:sampleupw}}
\label{fig:canonical}
\end{figure}

\begin{theorem}\label{thm:sampleupw}
The sample fitting problem is \NP{}-complete for $U$-samples.
\end{theorem}
\begin{proof}
The membership in \NP{} follows from the following observation: if $n$ is greater than the total length of the samples, then return yes as the tree-like solution works. Otherwise, non-deterministically pick an automaton of the size $n$ and 
check whether it fits the sample. 

The $\NP$-hardness proof is inspired by the construction from~\cite[Theorem 6.2.1]{GrammaticalInference}.
For a given instance of the \satz{} problem $\varphi = \bigwedge_{i=1}^n C_i$ over variables $x_1, \ldots, x_n$ (not all variables need to occur in $\varphi$), 
we construct a $U$-sample $S_{\varphi}$ such that there is an automaton with $n$ states fitting $S_{\varphi}$ 
if and only if $\varphi$ is satisfiable. 

We fix the alphabet $\set{a_i,c_i, d_i \mid i=1, \ldots,n } \cup \set{b,t}$.
The sample $S_{\varphi}$ consists of:
\begin{enumerate}[S1]
\item\label{i:sampleo} $(c_i, a_j, x)$  for each $i, j \in \set{1, \dots, n}$, where $x$ is $1$ if $i=j$, and $0$ otherwise.
\item\label{i:sampleoo} $(c_i, d_j, x)$ for each $i, j \in \set{1, \dots, n}$, where $x$ is $1$ if $x_i$ is in $C_j$, and $0$ otherwise.
\item\label{i:sampleooo} $(c_i b, d_i, 1)$  for each $i \in \set{1, \dots, n}$.
\item\label{i:sampleVar} $(c_i b, t, x)$  for each $i \in \set{1, \dots, n}$, where $x$ is $1$ if the clause $C_i$ contains only positive literals and $0$ if it contains only negative literals.
\end{enumerate}

Assume that $\varphi$  is satisfiable and let $\sigma : \set{x_1, \dots, x_n} \to \set{0, 1}$ be 
a satisfying valuation. Then, we construct an automaton $\aut_{\varphi}$ 
consistent with the sample $S_{\varphi}$ starting from the structure presented in Figure~\ref{fig:canonical}. 
Then, we add the following transitions:
\begin{itemize}
\item for each $i$, a loop in $q_i$ on the letter $t$ with the value $\sigma(x_i)$,
\item for each $i, j$, a loop in $q_i$ on the letter $d_j$ with the value $1$ if $x_i$ is in $C_j$ and $0$ otherwise, 
\item for each clause $C_i$, if $C_i$ is satisfied because of a variable $x_j$, then we add a transition from $q_i$ to $q_j$ over $b$ (if there are multiple possible variables, we choose any).
\end{itemize}

The remaining transitions can be set in arbitrary way. 
The obtained automaton $\aut_{\varphi}$ is consistent with the sample  $S_{\varphi}$.

Now assume that there is an automaton $\aut$ of $n$ states, which is consistent with the sample.
We show that the valuation $\sigma$ such that $\sigma(x_i) = \valueL{\aut}(c_it^\omega) $ satisfies $\varphi$.
Let $q_i$ be the state where the automaton $\aut$ is after reading the word $c_i$. By \Sref{i:sampleo}, all the states $q_1, \dots, q_n$ are pairwise different. Since there are only $n$ states, $q_1, \dots, q_n$ are all the states of $\aut$.
Now consider any clause $C_i$. Let $q_j$ be the state of $\aut$ after reading $c_ib$. Notice that by \Sref{i:sampleooo}, the value of $d_i^\omega$ in $q_j$ is 1, and by \Sref{i:sampleoo}, this means that $x_j$ is in $C_i$. 
If $C_i$ contains only positive literals, then the value of $t^\omega$ in $q_j$ is $1$ by \Sref{i:sampleVar}, which means that $\sigma(x_j)=1$ and that $C_i$ is satisfied. The other case is symmetric.
\end{proof}

\begin{theorem}\label{thm:samplee}
The sample fitting problem is \NP{}-complete for $E$-samples.
\end{theorem}

\begin{proof}
The proof is similar to the proof of Theorem~\ref{thm:sampleupw}. For the \NP{}-hardness, the sample now is obtained from the sample in the proof of Theorem~\ref{thm:sampleupw} by replacing every triple $(u, v, x)$ by the pair $(uv, x)$. However, now we ask for an automaton of size $n+2$.
If there is a valuation that satisfies a given set of clauses, then one can construct an automaton based on the one presented in Figure~\ref{fig:canonicaltwo}.

On the other hand, if there is an automaton fitting the sample, then it has to have a state where the expected value of any word is $0$, a state where the expected value of any word is $1$, and $n$ different states reachable by each $c_1, 
\dots, c_n$. The rest of the proof is virtually the same as in Theorem~\ref{thm:sampleupw}, except that now we define $\sigma$ such that 
$\sigma(x_i)$ is the expected value of words with the prefix $c_it$.
\begin{figure}
\begin{tikzpicture}[state/.style={draw,circle}]
\newdimen\nodeDist
\nodeDist=6mm

\node[state] (c1) at (0,0) {$q_1$};
\node[state] (c2) at (4,0) {$q_2$};
\node[state] (cn) at (10,0) {$q_n$};
\node (dots) at (7, 0) {\dots};

\node[state] (false) at (5,-1.2) {$q_F$};

\node[state] (true) at (5,1.5) {$q_T$};

\draw[->] (-0.6, 0.6) to (c1);

\draw[->] (c1) to node[below]{$c_2$} (c2);
\draw[->] (c1) to[bend left,in=170,out=10] node[above]{$c_n$} (cn);

\draw[->] (c1) to node[above]{$a_1$} (true);
\draw[->] (c1) to node[below=0.12em,text width=4cm]{$a_2$ \dots $a_n$} (false);

\draw[->] (c2) to node[left]{$a_2$} (true);
\draw[->] (c2) to node[right,text width=4cm]{$a_1$ $a_3$ \dots $a_n$} (false);

\draw[->] (cn) to node[above]{$a_n$} (true);
\draw[->] (cn) to node[below right=0.2em,text width=2.6cm]{$a_1$ \dots $a_{n-1}$} (false);

\draw[->,loop left] (c1) to node[left]{$c_1$} (c1);

\draw[->,loop left] (true) to node[left=-1em,text width=1cm] {$*$: 1} (true);

\draw[->,loop left] (false) to node[text width=1cm,left=-1em] {$*$} (false);

\end{tikzpicture}
\caption{A picture of the canonical automaton from the proof of Theorem~\ref{thm:samplee}. All the weights are $0$ except from transitions from the state $q_T$}\label{fig:canonicaltwo}
\end{figure}
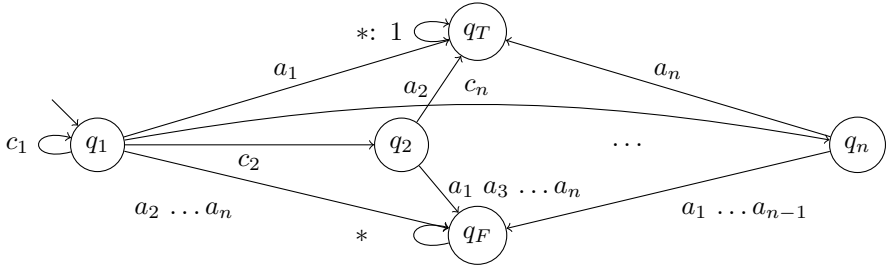
\end{proof}

The above proofs also work with some natural relaxations of the sample fitting problem.
 For example, if we only require the automaton to fit the samples up to some $\epsilon < \frac{1}{2}$, then the proofs still hold since we use only weights $0$ and $1$. 
Another relaxation for the $E$-samples case is to allow the automaton to give wrong values for some samples as long as the summarized probability of the examples with wrong value is less than some $\epsilon$.
However, since all the words are of the length at most three, the probability of $u \Sigma^{\omega}$ for each example $u$ is greater than 
$\frac{1}{(3n+2)^4}$ (recall that $|\Sigma| = 3n+2$), which means that for any $\epsilon<\frac{1}{(3n+2)^4}$ for every example some its extension 
must fit and hence the whole sample must fit.
 
\section{Active learning}
\label{sec:active}

In the active case, the learning algorithm can ask queries to an oracle, which is called \emph{the teacher}, which has a (hidden) function $\lang \colon \Sigma^{\omega} \to \R$ and answers two types of queries: 
\begin{itemize}
\item \emph{expectation queries}: given a finite word $u$, the teacher returns $\expected(\lang \mid u\Sigma^{\omega})$, 
\item \emph{$\epsilon$-consistency queries}: given an automaton $\aut$, if 
$\valueL{\aut}$  $\epsilon$-approximates $\lang$ (i.e., $\expected(|\lang - \valueL{\aut}|) \leq \epsilon$), the teacher returns YES, otherwise the teacher returns
a word $u$ such that $|\expected(\lang \mid u\Sigma^{\omega}) - \expected(\valueL{\aut} \mid u\Sigma^{\omega})| > \epsilon$. 
\end{itemize}

\noindent\emph{Remark}. 
Consider functions $\lang_1, \lang_2$ defined by $\flimavg$-automata.
If $\expected(|\lang_1 - \lang_2|) > \epsilon$, then due to Theorem~\ref{f:ExpectedLIMAVG}
there is a word $u\Sigma^{\omega}$ such that $\expected(|\lang_1 - \lang_2| \mid u\Sigma^{\omega}) > \epsilon$ and 
$\lang_1$ (resp., $\lang_2$) returns $\expected(\lang_1  \mid u\Sigma^{\omega})$ (resp,. $\expected(\lang_2  \mid u\Sigma^{\omega})$) 
on almost all words from $u\Sigma^{\omega}$. Therefore, 
$|\expected(\lang_1 \mid u\Sigma^{\omega}) - \expected(\lang_2 \mid u\Sigma^{\omega})| = \expected(|\lang_1 - \lang_2| \mid u\Sigma^{\omega}) > \epsilon$.

In the active learning case we consider two problems: approximate learning and rigid approximate learning. We first define approximate learning:

\begin{problem}[Approximate learning]
Given $\epsilon \in \Q^+ \cup \{0\}$ and a teacher with a (hidden) function $\lang$, 
construct a $\flimavg$-automaton $\aut$ such that $\valueL{\aut}$ $\epsilon$-approximates $\lang$ and
$\aut$ has the minimal number of state among such automata.
\end{problem}

\newcommand{\Pre}{\mathcal{Q}}
\newcommand{\Suff}{\mathcal{T}}

\newcommand{\neibr}{\mathbf{N}^G}
\newcommand{\lOneNorm}[1]{\left\lVert#1\right\rVert_{1}}

\newcommand{\C}{\mathcal{C}}
\newcommand{\D}{\mathcal{D}}

\subsection{Approximate learning}
We define a decision problem, called approximate minimization, which can be solved in polynomial-time
having a polynomial-time approximate learning algorithm. 

\begin{problem}[Approximate minimization]
Given a $\flimavg$-automaton $\aut$, $n \in \N$ and $\epsilon \in \Q^+$, \emph{the approximate minimization problem} asks whether 
there exists a $\flimavg$-automaton $\aut'$ with
at most $n$ states such that $\expected(|\valueL{\aut} - \valueL{\aut'}|) \leq \epsilon$.
\end{problem}

An efficient learning algorithm can be used to efficiently compute approximate minimization of a given $\flimavg$-automaton $\aut$;
we can run it and compute answers to queries of the learning algorithm in polynomial time in $|\aut|$~\cite{BaierBook}.
We show that the approximate minimization problem is 
$\NP$-complete, which means that approximate learning cannot be done in polynomial time if $P \neq NP$.

\begin{restatable}{theorem}{MinimalizationNPComp}
The approximate minimization problem is $\NP$-complete.
\label{th:approximateNPComp}
\end{restatable}
\begin{proof}[Proof sketch]
The problem is contained in $\NP$ as we can non-deterministically pick an automaton with $n$ states and check whether it $\epsilon$-approximates $\aut$. 

For a vector $\vec{v} \in \R^m$ we define $\lOneNorm{\vec{v}} = \sum_{i=1}^m |\vec{v}[i]|$.
For $\NP$-hardness, consider the following problem:
\emph {Binary k-Median Problem (BKMP)}: given 
	numbers $n,m, k, t \in \N$ 
	and a set of Boolean vectors $\C = \{ \vec{v}_1, \ldots, \vec{v}_n \} \subseteq \{0,1\}^m$, 
decide whether there are vectors $\vec{u}_1, \ldots, \vec{u}_k \in \{0,1\}^m$ and a partition $\D_1, \ldots, \D_k $ of $\C$ such that 
	$\sum_{j=1}^k \sum_{\vec{v} \in \D_j} \lOneNorm{\vec{v} - \vec{u}_j} \leq t$?

BKMP has been shown $\NP$-complete in~\cite{BKMP}. To ease the reduction, we consider a variant of BKMP, called Modified BKMP, 
where we assume that $\vec{0},\vec{1}$ belong to the instance and the solution and some additional constraints.  
The Modified BKMP is also \NP-complete.

For $\C = \{ \vec{v}_1, \ldots, \vec{v}_n\} \subseteq \{0,1\}^{m}$ we define $\lang_{\C}$ over the alphabet $\Sigma = \{a_1, \ldots, a_{m}\}$ as follows:
for all $a_i, a_j \in \Sigma$, 
$w \in \Sigma^\omega$ 
\begin{itemize}
\item if $i \leq n$, we set $\lang_{\C}(a_ia_j w) = \vec{v}_i[j]$. 
\item if $i \in \{n + 1, m\}$, we set $\lang_{\C}(a_ia_j w) = \vec{v}_{n}[j]$. 
\end{itemize}
Intuitively, we select the vector with the first letter and the vector's component with the second letter.
This language can be defined with a tree-like $\flimavg$-automaton $\aut_{\C}$. 

We show that, if an instance $(\C,k,t)$ of Modified BKMP  has a solution $\vec{u}_1, \ldots, \vec{u}_k$ and  $\D_1, \ldots, \D_k$, then there exists an automaton $\aut$ with $k+1$ states that $\frac{t}{nm}$-approximates $\lang_{\C}$. Let $\vec{v}_1 = \vec{u}_1 = \vec{0}$ and
$\vec{v}_2 = \vec{u}_2 = \vec{1}$.
The automaton $\aut$ consists of the initial state $q_0$ and the successors of the initial state, $q_1, \ldots, q_k$, which correspond to vectors $\vec{u}_1, \ldots, \vec{u}_k$, i.e., 
$q_1$ is a bottom SCC of the value $0$, $q_2$ is a bottom SCC of the value $1$, and for $i > 2$
the successors of $q_i$ encode $\vec{u}_i$ via $\delta(q_i, a_j) = q_1$ if $\vec{u}_i[j] = 0$
and $\delta(q_i, a_j) = q_2$ otherwise. 
The successors of $q_0$ are defined based on the partition $\D_1, \ldots, \D_k $, i.e., if $\vec{v}_i$ belongs to $C_j$, then $\delta(q_0,a_i) = q_j$. Observe that $\aut$ $\epsilon$-approximates $\aut_{\C}$.

Conversely, consider $\aut'$ with $k+1$ states that $\frac{t}{nm}$-approximates $\lang_{\C}$.
We define vectors $\vec{p}_1, \ldots, \vec{p}_m \in \R^{m}$ such that $\vec{p}_i[j] = \expected(\aut' \mid a_ia_j \Sigma^{\omega})$. 
The structure of Modified BKMP, implies that 
the initial state of $\aut'$ has no self loops and hence it has at most $k$ different successor states.
Therefore, there are at most $k$ different vectors among $\vec{p}_1, \ldots, \vec{p}_m$.
Finally, we observe that since we consider $\lOneNorm{\cdot}$, w.l.o.g. we can assume that $\vec{p}_1, \ldots, \vec{p}_m \in \{0,1\}^m$.
Therefore, these vectors give us a solution to the instance $(\C,k,t)$ of Modified BKMP.
\end{proof}

\subsection{Rigid approximate learning}
One of the drawbacks of the standard approximation is that the counterexamples may be dubious, if not useless. We illustrate this with an example.

\begin{example}[Dubious counterexamples]
\label{ex:ambiguous}
\newcommand{\autDFA}{\mathcal{B}}
Consider a minimal DFA $\autDFA$ with $2^n$ states whose language consists of words of length $m$ for some $n < m$, and a word $v \in \Sigma^{n}$.
We define a function $\lang_{\autDFA,v} \colon \set{a,b}^{\omega} \to \R$ such that for all $u,w$
\begin{itemize}
\item $\lang_{\autDFA,v}(auw)=0$ if $|u| = n$ and $\autDFA$ accepts $u$, 
\item $\lang_{\autDFA,v}(auw)=0.3$ if $|u| = n$ and $\autDFA$ rejects $u$, 
\item $\lang_{\autDFA,v}(bw)$ is $0.4$ if the first occurrence of $v$ in $w$ is followed by $a$, and $1$ otherwise.
\end{itemize}

Fix $\epsilon=0.1$.
Observe that 
$\lang_{\autDFA,v}$ can be $0.1$-approximated with an automaton $\aut$, which is faithful to $\lang_{\autDFA,v}$ on $b\Sigma^{\omega}$ and returns $0.15$ for all other words. 
$\aut$ has $n + O(1)$ states.

Assume that the teacher gives only counterexamples starting with $a$, and hence $\epsilon$-consistency queries do not give any information about 
the values of words starting with $b$. The teacher can do it as long as the algorithm does not know the whole $\autDFA$, which takes $\Omega(|\autDFA|)$ queries to learn.
Yet even if the algorithm learns the whole $\autDFA$ and returns the $0.7$ for the words starting with $b$, the expected difference is $0.5 \cdot 0.3 = 0.15$. 
It follows that to learn the approximation, the algorithm needs to learn something about $v$. 

Suppose that the algorithm did not learn the whole $\autDFA$. Then, to learn something non-trivial about words starting with $b$, 
it has to ask an expectation query containing $v$. Since the learning algorithm is deterministic, 
it asks the same expectation queries for different words $v$.
Therefore, for every learning algorithm there are words $v$ that can be learned only after asking queries
of the total length  $2^{\Omega(|v|)}$. 

It follows that any learning algorithm has to ask queries of total length $\Omega(|\autDFA|)$ or 
$2^{\Omega(|v|)}$,  which totals to $2^{\Omega(n)}$.
\end{example}

In Example~\ref{ex:ambiguous} we assumed an antagonistic teacher, which misleads the algorithm on purpose.
But even with a stochastic teacher, 
it is not known whether ``fixing'' a given random counterexample is a step towards better approximation. 
To resolve this issue, we consider a stronger notion of approximation, called \emph{rigid approximation}, where 
we require all conditional expected values to be $\epsilon$-close, i.e., for all words $u \in \Sigma^*$  
we have $\expected(|\lang_1 - \lang_2| \mid u\Sigma^{\omega}) \leq \epsilon$.
In this framework counterexamples are certain, i.e., if for some $u \in \Sigma^*$, the expectation over $u \Sigma^{\omega}$ is more than $\epsilon$ off the intended value, 
it has to be modified. Formally, we define the problem as follows:

\begin{problem}[Rigid approximate learning]
Given $\epsilon \in \Q^+$ and a teacher with a (hidden) function $\lang$, 
construct an automaton $\aut$  
such that $\valueL{\aut}$ is a rigid $\epsilon$-approximation of $\lang$ and
$\aut$ has the minimal number of state among such automata.
\end{problem}

Even though the counterexamples are certain in this framework, we observe that this does not eliminate ambiguity.
For instance, there can be multiple automata with the minimal number of states.

\begin{example}[Non-unique minimalization]
Consider the automaton $\aut_1$ depicted in Figure~\ref{fig:manyMinimalAutomata} and $\epsilon = \frac{1}{4}$. 
Any automaton $\aut$, which is a rigid $\epsilon$-approximation of $\valueL{\aut_1}$, has at least two bottom SCCs and hence it requires at least 3 states.
Therefore the automata $\aut_2, \aut_3$ depicted in Figure~\ref{fig:manyMinimalAutomata}, which $\epsilon$-approximate $\valueL{\aut_1}$, have the minimal number of states.
This shows that there can be multiple correct answers to in the rigid approximate learning problem. 

Based on this example, we construct the automaton $\aut_{\textrm{exp}}$ parametrized by $n \in \N$ depicted in Figure~\ref{fig:manyMinimalAutomata},  which has $O(n)$ states and there are 
exponentially many (in $n$) minimal non-equivalent automata, which are rigid $\frac{1}{8n}$-approximations of $\valueL{\aut_{\textrm{exp}}}$.
\begin{figure}
\begin{tikzpicture}
[state/.style={draw,circle}]
\newdimen\nodeDist
\nodeDist=6mm

\node[state] (Src) at (0,0) {};
\node[state] (A) at (-0.9,-1) {};
\node[state] (B) at (0,-1) {};
\node[state] (C) at (0.9,-1){};

\node at (-0.7,0) {$\aut_1$:};

\draw[->] (Src) to node[left]{a} (A);
\draw[->] (Src) to node[left]{b} (B);
\draw[->] (Src) to node[left]{c} (C);

\draw[->,loop below] (A) to node[below] {0} (A);
\draw[->,loop below] (B) to node[below] {0.5} (A);
\draw[->,loop below] (C) to node[below] {1} (A);

\begin{scope}[xshift=3.0cm,,yshift=0.5cm]
\node at (-0.7,-0.0) {$\aut_2$:};

\node[state] (Src) at (0,0) {};
\node[state] (A1) at (-0.6,-1) {};
\node[state] (B1) at (0.6,-1) {};

\draw[->,bend right] (Src) to node[left]{a} (A1);
\draw[->,bend left] (Src) to node[right]{b} (A1);
\draw[->] (Src) to node[right]{c} (B1);

\draw[->,loop below] (A1) to node[left] {0.25} (A);
\draw[->,loop below] (B1) to node[left] {1} (A);
\end{scope}

\begin{scope}[xshift=3.0cm,yshift=-1.5cm]
\node at (-0.7,-0.1) {$\aut_3$:};
\node[state] (Src) at (0,0) {};
\node[state] (A2) at (-0.6,-1) {};
\node[state] (B2) at (0.6,-1) {};

\draw[->] (Src) to node[left]{a} (A2);
\draw[->,bend right] (Src) to node[left]{b} (B2);
\draw[->,bend left] (Src) to node[right]{c} (B2);

\draw[->,loop below] (A2) to node[left] {0} (A);
\draw[->,loop below] (B2) to node[left] {0.75} (A);
\end{scope}

\begin{scope}[xshift=9cm,yshift=+1.6cm]

\node at (-1.4,-1.2) {$\aut_{\textrm{exp}}$:};

\draw (0,-1) -- (3.2,-2.5) -- (-2.8,-2.5) -- cycle;

\node at (1.2,-3.5) {$\dots$};

\begin{scope}[xshift=-2.8cm,yshift=-0.5cm]
\node[state] (Src) at (0,-2.2) {};
\node[state] (A3) at (-0.9,-3.2) {};
\node[state] (B3) at (0,-3.2) {};
\node[state] (C3) at (0.9,-3.2){};

\draw[->] (Src) to node[left]{a} (A3);
\draw[->] (Src) to node[left]{b} (B3);
\draw[->] (Src) to node[left]{c} (C3);

\draw[->,loop below] (A3) to node[below] {0} (A);
\draw[->,loop below] (B3) to node[below] {$\frac{1}{4n}$} (A);
\draw[->,loop below] (C3) to node[below] {$\frac{2}{4n}$} (A);
\end{scope}

\begin{scope}[xshift=-0.6cm,yshift=-0.5cm]
\node[state] (Src) at (0,-2.2) {};
\node[state] (A3) at (-0.9,-3.2) {};
\node[state] (B3) at (0,-3.2) {};
\node[state] (C3) at (0.9,-3.2){};

\draw[->] (Src) to node[left]{a} (A3);
\draw[->] (Src) to node[left]{b} (B3);
\draw[->] (Src) to node[left]{c} (C3);

\draw[->,loop below] (A3) to node[below] {$\frac{4}{4n}$} (A);
\draw[->,loop below] (B3) to node[below] {$\frac{5}{4n}$} (A);
\draw[->,loop below] (C3) to node[below] {$\frac{6}{4n}$} (A);
\end{scope}

\begin{scope}[xshift=3.2cm,yshift=-0.5cm]
\node[state] (Src) at (0,-2.2) {};
\node[state] (A3) at (-0.9,-3.2) {};
\node[state] (B3) at (0,-3.2) {};
\node[state] (C3) at (0.9,-3.2){};

\draw[->] (Src) to node[left]{a} (A3);
\draw[->] (Src) to node[left]{b} (B3);
\draw[->] (Src) to node[left]{c} (C3);

\draw[->,loop below] (A3) to node[below] {$\frac{4n-4}{4n}$} (A);
\draw[->,loop below] (B3) to node[below] {$\frac{4n-3}{4n}$} (A);
\draw[->,loop below] (C3) to node[below] {$\frac{4n-2}{4n}$} (A);
\end{scope}

\end{scope}
\end{tikzpicture}
 \caption{The automaton $\aut_1$ approximated by two minimal non-equivalent automata and the automaton $\aut_{\textrm{exp}}$ approximated by exponentially many non-equivalent automata}
\label{fig:manyMinimalAutomata}
\end{figure}
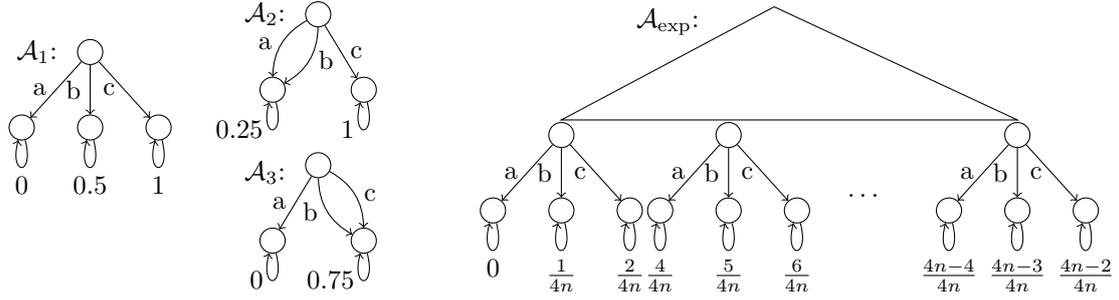
\end{example}

As in the approximate learning case, efficient rigid approximate learning enables us to solve efficiently the following \emph{rigid approximate minimization} problem.

\begin{problem}[Rigid approximate minimization]
Given a $\flimavg$-automaton $\aut$, $n \in \N$ and $\epsilon \in \Q^+$, construct a $\flimavg$-automaton $\aut'$ with
at most $n$ states such that 
for all words $u \in \Sigma^*$ we have $\expected(|\valueL{\aut} - \valueL{\aut'}| \mid u\Sigma^{\omega}) \leq \epsilon$.
\end{problem}

First, we consider a naive approach to solve the rigid approximate minimization problem based on state merging. 
We start with an input automaton $\aut$, merge its states to maintain the property that the automaton with merged states $\aut'$
rigidly $\epsilon$-approximates $\valueL{\aut}$. We terminate if the automaton is \emph{minimal}, i.e., merging any two states of $\aut'$ violates the property.
However, it may happen that $q_1$ can be merged either with $q_2$ or $q_3$, but not with both. 
We show in the following example that the choice of merged states can have profound impact on the size of a minimal automaton.

\begin{example}[Minimal automata of different size]
Assume $n \in \N$ and the function $\lang$ that returns $0$ on words from $a\Sigma^na\Sigma^{\omega}$, $1$ on $b\Sigma^nb\Sigma^{\omega}$, and $0.5$ otherwise.
Let $\aut$ be a minimal automaton defining such a language, as depicted in Figure~\ref{fig:minimalAutomataDifferentSize}.

Let $\epsilon=\frac{1}{4}$. 
To minimize $\aut$, we can merge it to the automaton $\aut_{S}$ with 3 states or to $\aut_{L}$ with $n+3$ states.
Observe that $\aut_{S}$ and $\aut_{L}$ are rigid $\epsilon$-approximations of $\valueL{\aut}$.
For 
$\aut_S$, it is because for every word $au$ we have $\expected(\aut_B \mid au\Sigma^{\omega})\in\set{0,0.5}$ and hence it is $\epsilon$-close to $0.25$, and similarly for $ba$.
For $\aut_L$, it is because for every $u$ of length at least $n+2$  the expected value of $\aut_L$ is $0.75$ or $0.25$ depending on the $(n+2)$th letter of $u$ whereas for $\aut$ it is in $\set{0.5, 1}$ or in $\set{0, 0.5}$, resp. 
For shorter $u$, we simply observe that the difference of expected values w.r.t. to a word does not exceed the maximal difference in its suffixes.

The automaton $\aut_L$ is minimal, as there are no states that can be merged. Therefore, the difference and even the ratio between the sizes of both automata are unbounded.
\begin{figure}
\begin{tikzpicture}
[state/.style={draw,circle}]
\newdimen\nodeDist
\nodeDist=6mm

\node (label1) at (-0.7,0) {$\aut$:};
\node (label2) at (6.3,1) {$\aut_L$:};
\node (label3) at (6.3,-1) {$\aut_S$:};

\node[state] (Src) at (0,0) {};

\node[state] (L1) at (1,1) {};
\node[state] (L2) at (2,1) {};
\node (dots) at (3, 1) {\dots};
\node[state] (LN) at (4,1) {};

\node[state] (R1) at (1,-1) {};
\node[state] (R2) at (2,-1) {};
\node (dots) at (3, -1) {\dots};
\node[state] (RN) at (4,-1) {};

\draw[->] (Src) to node[above] {a} (L1);
\draw[->] (L1) to node[above] {*} (L2);
\draw[->] (L2) to node[above] {*} (2.7, 1);
\draw[->] (3.3, 1) to node[above] {*} (LN);

\draw[->] (Src) to node[above] {b} (R1);
\draw[->] (R1) to node[above] {*} (R2);
\draw[->] (R2) to node[above] {*} (2.7, -1);
\draw[->] (3.3, -1) to node[above] {*} (RN);

\node[state] (V1) at (5,1.5) {};
\node[state] (V2) at (5,0.5) {};
\node[state] (V3) at (5,-0.5) {};
\node[state] (V4) at (5,-1.5) {};

\draw[->,loop below] (V1) to node[right] {0} (V1);
\draw[->,loop below] (V2) to node[right] {0.5} (V2);
\draw[->,loop below] (V3) to node[right] {0.5} (V3);
\draw[->,loop below] (V4) to node[right] {1} (V4);

\draw[->] (LN) to node[above] {a} (V1);
\draw[->] (LN) to node[above] {b} (V2);
\draw[->] (RN) to node[above] {a} (V3);
\draw[->] (RN) to node[above] {b} (V4);

\begin{scope}[xshift=7cm,yshift=1cm]
\node[state] (Src) at (0,0) {};

\node[state] (L1) at (1,0) {};
\node[state] (L2) at (2,0) {};
\node[state] (LN) at (4,0) {};

\draw[->] (Src) to node[above] {*} (L1);
\draw[->] (L1) to node[above] {*} (L2);
\draw[->] (L2) to node[above] {*} (2.7, 0);
\draw[->] (3.3, 0) to node[above] {*} (LN);

\node[state] (V1) at (5, 0.5) {};
\node[state] (V2) at (5, -0.5) {};

\draw[->,loop below] (V1) to node[right] {0.25} (V1);
\draw[->,loop below] (V2) to node[right] {0.75} (V2);

\draw[->] (LN) to node[above] {a} (V1);
\draw[->] (LN) to node[above] {b} (V2);

\end{scope}

\begin{scope}[xshift=7cm,yshift=-1cm]

\node[state] (Src3) at (0,-0.) {};

\node[state] (A1) at (1, 0.5) {};
\node[state] (A2) at (1, -0.5) {};

\draw[->] (Src3) to node[above] {a} (A1);
\draw[->] (Src3) to node[above] {b} (A2);

\draw[->,loop below] (A1) to node[right] {0.25} (V13);
\draw[->,loop below] (A2) to node[right] {0.75} (V23);

\end{scope}

\end{tikzpicture}

 \caption{Minimal automata of different size}
\label{fig:minimalAutomataDifferentSize}
\end{figure}
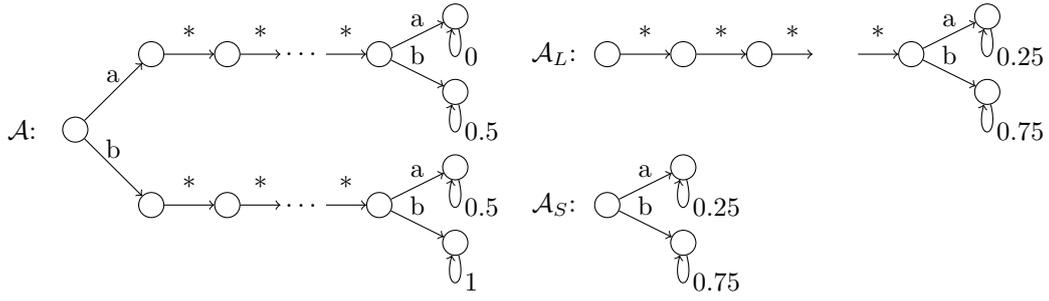
\end{example}

We show that the rigid approximate minimization problem is $\NP$-complete, which implies that there is no polynomial-time 
rigid approximate learning algorithm (unless $P = NP$). 

\begin{restatable}{theorem}{MinimalRigidNPComp}
The rigid approximate minimization problem is $\NP$-complete.
\end{restatable}
\begin{proof}[Proof sketch]
The rigid approximate minimization problem is in $\NP$. We show that it is $\NP$-hard.
We define a problem, which is an intermediate step in our reduction. 

Given $n,k \in \N$ and vectors $\C = \set{\vec{v}_1, \ldots, \vec{v}_n } \subseteq \{0, \frac{1}{2}, 1\}^n$, \emph{the $\frac{1}{4}$-vector cover problem} asks
whether there exist $\vec{u}_1, \ldots, \vec{u}_k \in \R^{n}$ such that for every vector $\vec{v} \in \C$ there is $j$ such that 
$\lInfNorm{\vec{v} - \vec{u}_j} \leq \frac{1}{4}$?

The $\frac{1}{4}$-vector cover problem is \NP-complete; the \NP-hardness is via reduction from the dominating set problem. 
We show that the $\frac{1}{4}$-vector cover problem reduces to the rigid approximate minimization problem

Let  $\C = \set{ \vec{v}_1, \ldots, \vec{v}_n } \subseteq \{0, \frac{1}{2}, 1\}^n$.
We define $\lang_{\C}$ over the alphabet $\Sigma = \{a_1, \ldots, a_{m}\}$ such that for all $a_i, a_j \in \Sigma$, 
$w \in \Sigma^\omega$ we have $\lang_{\C}(a_ia_j w) = \vec{v}_i[j]$. Such $\lang_{\C}$ can be defined with a 
a tree-like $\flimavg$-automaton $\aut_{\C}$.

Let $\epsilon = \frac{1}{4}$. We show that an instance $n,k,\C$ 
of the $\frac{1}{4}$-vector cover problem has a solution  if and only if $\aut_{\C}$ has a rigid $\epsilon$-approximation with $k+3$ states.

Assume that there is $\aut'$ with $k+3$ states such that for all words $u \in \Sigma^*$ we have $\expected(|\valueL{\aut_{\C}} - \valueL{\aut'}| \mid u\Sigma^{\omega}) \leq \epsilon$.
Observe that the initial state of $\aut'$ has at most $k$ different successors.
To see that consider functions $\lang^{u}(w) = \valueL{\aut_{\C}}(u w)$ defined for $u \in \Sigma^*$. 
If for some $u_1, u_2, u \in \Sigma^*$ we have $\expected(|\lang^{u_1 u} - \lang^{u_2 u}|) > 2 \epsilon$, then words $u_1$ and $u_2$ lead (from the initial state) to different states in $\aut'$.
It follows that $\aut'$ has at least $3$ states that are no successors of the initial state: the initial state and
(at least two) states that correspond to bottom SCCs.
We define vectors $\vec{u}_1, \ldots, \vec{u}_k$ from the successors of the initial state of $\aut'$. Formally,
for $i \in \{1, \ldots, n\}$ we define a vector $\vec{y}_i$ as $\vec{y}_i[j] = \expected(\valueL{\aut'} \mid a_ia_j \Sigma^{\omega})$. 
There are at most $k$ different vectors among $\vec{y}_1, \ldots, \vec{y}_n$ and we take these distinct vectors as $\vec{u}_1, \ldots, \vec{u}_k$.
The condition $\expected(|\valueL{\aut_{\C}} - \valueL{\aut'}| \mid u\Sigma^{\omega}) \leq \epsilon$ implies that 
$\vec{u}_1, \ldots, \vec{u}_k$ form a solution to $n,k,\C$.

Conversely, assume that the instance $n,k,\C$ has a solution. Observe that w.l.o.g. we can assume that
the solution vectors $\vec{u}_1, \ldots, \vec{u}_k$ belong to $\set{0.25,0.75}^n$.
Based on this solution we define $\aut'$ with $k+3$ states 
$q_0, q_1, \ldots, q_k, s_1, s_2$ such that $q_0$ is initial, $q_1, \ldots, q_k$ are successors of $q_0$, and $s_1, s_2$ are single-state bottom SCC  of the values $0.25$ and $0.75$.
The states $q_1, \ldots, q_k$ correspond to vectors $\vec{u}_1, \ldots, \vec{u}_k$, i.e., the successors of $q_i$ encode $\vec{u}_i$ via $\delta(q_i, a_j) = s_1$ if $\vec{u}_i[j] = 0.25$
and $\delta(q_i, a_j) = s_2$ otherwise. 
The successors of $q_0$ are defined based on the matching from the vector cover, i.e., if $\delta(q_0,a_i) = q_j$, then $\vec{u}_{j}$ is $\frac{1}{4}$-close to $\vec{v}_i$.
Observe that $\aut'$ is a rigid $\frac{1}{4}$-approximation of ${\aut_{\C}}$.

\end{proof}

\newcommand{\eqAlmostSure}{\equiv^0_\lang}

\section{Almost-exact learning}
The almost-exact learning the minimal automaton is defined as follows.

\begin{problem}[Almost-exact learning]
Given a teacher with a (hidden) function $\lang$, 
construct a $\flimavg$-automaton $\aut$ such that $\valueL{\aut}$ is a $0$-approximation of $\lang$.
\end{problem}

Notice that for functions $\lang$ and $\valueL{\aut}$, the following conditions are equivalent:
\begin{AutoMultiColItemize}
\item $\valueL{\aut}$ is a $0$-approximation of $\lang$.
\item $\valueL{\aut}$ is a rigid $0$-approximation of $\lang$.
\item $\prob(\set{ w : \valueL{\aut}(w) \neq \lang(w)}){=}0$.
\item $\prob(\set{ w : \valueL{\aut}(uw) \neq \lang(uw)}){=}0$ for each $u$.
\end{AutoMultiColItemize}

We show that there is a polynomial-time algorithm for almost-exact learning. 
\begin{theorem}
The almost-exact learning problem for $\flimavg$-automata can be solved in polynomial time in the size of the minimal automaton that is almost equivalent with the target function $\lang$ and 
the maximal length of counterexamples returned by the teacher. 
\label{th:almostSure}
\end{theorem}

\begin{proof}
We define a relation $\eqAlmostSure$ on finite words $u,v \in \Sigma^*$ as follows
\[ u \eqAlmostSure v \text{ if and only if }
\prob(\set{ w : \lang(uw) \neq \lang(vw)}) = 0
\]

A quick check shows that the relation $\eqAlmostSure$ is an equivalence relation.
We show that $\eqAlmostSure$ is a right congruence, i.e., if $u \eqAlmostSure v$, then for all $a$ we have $ua \eqAlmostSure va$.
Indeed, consider $X_1 = \set{ w : \lang(uaw) \neq \lang(vaw)}$ and $X_2 = \set{ w : \lang(uw) \neq \lang(vw)}$.
Note that $u \eqAlmostSure v$ implies $\prob(X_2) = 0$.
For all $w$, if $w \in X_1$, then $a w \in X_2$. It follows that (under the uniform distribution) 
$\prob(X_1) \leq \frac{1}{|\Sigma|} \prob(X_2) = 0$. Thus, $\prob(X_1) =0$ and $ua \eqAlmostSure va$.

We now show a counterpart of the Myhill–Nerode theorem: there is a $\flimavg$-automaton of $n$ states defining almost-exactly $\lang$ if and only if the index of $\eqAlmostSure$ is $n$. 

If $\aut$ defines almost-exactly $\lang$, then the index of $\eqAlmostSure$ is bounded by the number of states of $\aut$.
Indeed, if for words $u,v$ the automaton $\aut$ starting from the initial state ends up in the same state, then
for all words $w$ we have $\lang(uw) =  \lang(vw)$ and hence $u \eqAlmostSure v$. 
Conversely, assume that $\eqAlmostSure$ has a finite index. Then, we construct a $\flimavg$-automaton $\aut_{\eqAlmostSure}$ corresponding to $\eqAlmostSure$.
The states of $\aut_{\eqAlmostSure}$ are the equivalence classes of $\eqAlmostSure$ and 
$\aut_{\eqAlmostSure}$ has a transition from $[u]_{\eqAlmostSure}$ to $[v]_{\eqAlmostSure}$ over $a$  if and only if 
$[ua]_{\eqAlmostSure} = [v]_{\eqAlmostSure}$. Observe that due to Theorem~\ref{f:ExpectedLIMAVG}, if $[u]_{\eqAlmostSure}$ and $[v]_{\eqAlmostSure}$ are in a bottom SCC, then
$[u]_{\eqAlmostSure} = [v]_{\eqAlmostSure}$.
Then, for every bottom SCCs $[u]_{\eqAlmostSure}$ we assign the value of all outgoing transitions (which are self-loops) to 
$\expected(\lang \mid u\Sigma^{\omega})$. For the remaining transitions we set the weights to $0$ (due to Theorem~\ref{f:ExpectedLIMAVG} these weights are irrelevant as they do not change any of the expected values).
Observe that $\aut_{\eqAlmostSure}$ computes $\lang$.

A classical result of~\cite{angluin1987learning} states that DFA can be learned in polynomial time using membership and equivalence queries. We adapt this result here.
The learning algorithm for $\flimavg$-automata maintains a pair $(\Pre,\Suff)$, where 
	$\Pre$, the set of access words, contains different representatives the right congruence relation, and
	$\Suff$, the set of test words, contains words that approximate the right congruence relation.
	 $\Suff$ defines the relation $\equiv_{\Suff}$ such that
$u_1 \equiv_{\Suff} u_2$ if and only if  
for all $v \in \Suff$ we have
 \( \prob (\set{w \mid  \lang(u_1vw) \neq \lang(u_2vw)}) = 0
\).

The algorithm maintains two properties: 
\emph{separability}: all different words $u_1, u_2 \in \Pre$ belong to different equivalence classes of $\equiv_{\Suff}$, and
\emph{closedness}; for all $u_1 \in \Pre$ and $a \in \Sigma$ there is $u_2 \in \Pre$ with $u_1 a \equiv_{\Suff} u_2$.
Each separable and closed pair $(\Pre,\Suff)$ defines a $\flimavg$-automaton $\aut_{\Pre, \Suff}$, which can be tested for $0$-consistency against the teacher's function.
If the teacher provides a counterexample $u$, then it is used to 
extend $\Pre$ and $\Suff$. 
To do so, we split $u$ into $v_1a, v_2$ such that $v_1a$ is the minimal prefix of $u$ such that $\expected(\lang \mid v_1\Sigma^{\omega}) \neq \expected(\aut_{\Pre, \Suff} \mid u\Sigma^{\omega})$. 
Let $v_1'$ be a word from $\Pre$ that is $\equiv_{\Suff}$-equivalent to $v_1$.
We take $\Pre' = \Pre \cup \set{v_1'a}, \Suff' = \Suff \cup \set{v_2}$ and close $(\Pre', \Suff')$ using expectation queries and test the equivalence again. 
We repeat this  until we get a $\flimavg$-automaton defining almost equivalent to $\lang$.

The proof of correctness of this algorithm is a straightforward modification of the proof of correctness of the algorithm from \cite{angluin1987learning}, and thus it is omitted.
\end{proof}

\subsection{Non-uniform distributions}
\label{sec:non-uniform}
So far we only discussed the uniform distribution of words. 
Here we briefly discuss whether Theorem~\ref{th:almostSure} can be generalized to arbitrary distributions, represented by Markov chains. 
We assume that the Markov chain is given to a learning algorithm in the input.
\todo{We assume that the Markov chain is given as an algorithm input; if not, it can be learned from the sample using standard techniques.
Cite the results. (Janek: I don't know such results.)}

\Paragraph{Markov chains} A (finite-state discrete-time) \emph{Markov chain} is a tuple $\tuple{\Sigma,S,s_0,E}$, 
where $\Sigma$ is the alphabet of letters, 
$S$ is a finite set of states, $s_0$ is an initial state, 
$E \colon S \times \Sigma \times S \mapsto [0,1]$ is an edge probability function, which  
for every $s \in S$ satisfies $\sum_{a \in \Sigma, s' \in S} E(s,a,s') = 1$. 

The probability of a finite word $u$ w.r.t. a Markov chain $\markov$, denoted by $\prob_{\markov}(u)$,  is the sum of probabilities of paths from $s_0$ labeled by $u$, 
where the probability of a path is the product of probabilities of its edges.
For basic open sets $u\cdot \Sigma^\omega = \{ uw \mid w \in \Sigma^{\omega} \}$, 
we have $\prob_{\markov}(u\cdot \Sigma^\omega)=\prob_{\markov}(u)$, and then the 
probability measure over infinite words defined by $\markov$ is the unique
extension of the above measure~\cite{feller}.
We will denote the unique probability measure defined by $\markov$ as $\prob_{\markov}$. 

Observe that the uniform distribution can be expressed with a (single-state) Markov chain and hence all the lower bounds from 
Section~\ref{sec:passive} and Section~\ref{sec:active} still hold. 

\Paragraph{Non-vanishing Markov chains} A Markov chain $\markov$ is \emph{non-vanishing} if for all words $u \in \Sigma^*$ we have $\prob_{\markov}(u\Sigma^*) > 0$.
The almost-exact active learning over distributions given by non-vanishing Markov chains can be solved in polynomial time using Theorem~\ref{th:almostSure}. For every measurable set $X$, non-vanishing Markov chain $\markov$  we have $\prob_{\markov}(X) > 0$ if and only if $\prob(X) > 0$.
Thus, almost-exact learning over $\markov$ and over the uniform distribution coincide.

\Paragraph{General Markov chains} The proof for the uniform distribution does not extend to vanishing Markov chains because the relation $\eqAlmostSure$ is not a right congruence.
This cannot be simply fixed  as we show that learning cannot be done in polynomial time, assuming P $\neq$ \NP{}. We define the almost-exact minimization problem as 
an instance of the approximate minimization problem with $\epsilon =0$. Having a polynomial-time algorithm for almost-exact learning, we can solve 
the almost-exact minimization problem in polynomial time.

\begin{theorem}
The almost-exact minimization problem for $\flimavg$-automata under distributions given by Markov chains is~\NP-complete. 
\end{theorem}
\begin{proof}
The problem is in \NP{} as we can non-deterministically pick an automaton $\aut'$ and check in polynomial time whether 
 $\flimavg$-automata $\aut$ and $\aut'$ are almost equivalent w.r.t. a given Markov chain.

We reduce the sample fitting problem, which is \NP-complete for $E$-samples (Theorem~\ref{thm:samplee}), to the almost-exact minimization problem.
Consider a finite $E$-sample $S$ based on words $u_1, \ldots, u_k$.
Let $\markov_S$ be a Markov chain which assigns probability $\frac{1}{k}$ to $u_i \Sigma^{\omega}$, for each $i$, and $0$ to words not starting with any of $u_i$.
On each $u_i \Sigma^{\omega}$, $\markov_S$ defines the uniform distribution.
Let $\aut_S$ be a a tree-like $\flimavg$-automaton consistent with $S$. Both, $\markov_{S}$ and $\aut_S$ are of polynomial size in $S$.
Observe that every automaton $\aut$, which is consistent with $S$ (according to the uniform distribution over infinite words), is almost equivalent to $\aut_S$ (over $\prob_{\markov}$) and vice versa.
Therefore, there is an automaton of $n$ states almost equivalent to $\aut_S$ (under the distribution $\prob_{\markov}$)  if and only if the sample fitting problem with $S$ and $n$ has a solution.
\end{proof}

\bibliography{papers}

\appendix

\section*{Appendix}
\section{NP{}-completeness of approximate minimization}
\label{app:appendixA}

\MinimalizationNPComp*
\begin{proof}
The problem is contained in $\NP$ as we can non-deterministically pick an automaton $\aut'$ with $n$ states and check whether it $\epsilon$-approximates $\aut$. It can be done in polynomial time in the following way. For all bottom SCCs $C$ from $\aut$ and $D$ from $\aut'$, we compute
the probability $p_{C,D}$ of the set of infinite words $w$ such that $\aut$ reaches $C$ over $w$ and $\aut'$ reaches $D$ over $w$. 
Let $x_C$ be the expected value of $C$ in $\aut$, which is attained by almost all words that reach $C$ and $y_D$ be the corresponding value for $D$ in $\aut'$. The values $p_{C,D}, x_C, y_D$ can be computed in polynomial time~\cite{BaierBook}.
Then, the expected difference $\expected(|\valueL{\aut} - \valueL{\aut'}|)$ is the sum over all bottom SCCs $C$ from $\aut$ and $D$ from $\aut'$ of
$p_{C,D} \cdot |x_C - y_D|$. 

For $\NP$-hardness, consider the following problem.
Given a vector $\vec{v} \in \R^m$, we define $\lOneNorm{\vec{v}} = \sum_{i=1}^m |\vec{v}[i]|$.
\medskip 

\emph {Binary k-Median Problem (BKMP)}: given 
	numbers $n,m, k, t \in \N$ 
	and a set of Boolean vectors $\C = \{ \vec{v}_1, \ldots, \vec{v}_n \} \subseteq \{0,1\}^m$, 
decide whether there are vectors $\vec{u}_1, \ldots, \vec{u}_k \in \{0,1\}^m$ and a partition $\D_1, \ldots, \D_k $ of $\C$ such that 
	$\sum_{j=1}^k \sum_{\vec{v} \in \D_j} \lOneNorm{\vec{v} - \vec{u}_j} \leq t$.
\medskip

BKMP has been shown $\NP$-complete in~\cite{BKMP}. To ease the reduction, we consider a variant of BKMP.

\emph{Modified Binary k-Median Problem (RBKMP)}: 
 given numbers $n,m,h, k, t \in \N$, with $h > max(2t, m,n)$, 
 and a set of Boolean vectors $\C = \{ \vec{v}_1, \ldots, \vec{v}_n \} \subseteq \{0,1\}^{2m+2h}$, which satisfies the following conditions:
\begin{enumerate}[(C1)]
\item $\C$ contains vectors $\vec{0}$ and $\vec{1}$,
\item \label{cond:equalZero} each vector $\vec{v} \in \C \setminus \{\vec{0}, \vec{1}\}$ has as many $0$'s as $1$'s,
\item \label{cond:alignment} for every vector $\vec{v} \in \C \setminus \{\vec{0}, \vec{1}\}$ we have $\vec{v}[2m+1] = \ldots = \vec{v}[2m+h] = 1$ and 
$\vec{v}[2m+h+1] = \ldots = \vec{v}[2m+2h] = 0$.
\end{enumerate}
decide whether there are 
vectors $\vec{u}_1, \ldots, \vec{u}_k \in \{0,1\}^{2m+2h}$, including $\vec{0}$ and $\vec{1}$ , and a partition $\D_1, \ldots, \D_k $ of $\C$ such that 
$\sum_{j=1}^k \sum_{\vec{v} \in \D_j} \lOneNorm{\vec{v} - \vec{u}_j} \leq t$.

We reduce BKMP to Modified BKMP.
Consider an instance $(\C, k,t)$ of BKMP. 
For every vector $\vec{v} \in \{0,1\}^m$ we define $\vec{v}^* \in \{0,1\}^{2m+2h}$ such that
\begin{enumerate}
\item $\vec{v}^*[1,m] = \vec{v}$, 
\item $\vec{v}^*[m+1,2m] = \vec{1} - \vec{v}$, 
\item $\vec{v}^*[2m+1,2m+h] = 1$ 
\item $\vec{v}^*[2m+h+1, 2m+2h] = 0$.
\end{enumerate}
Then, we define $\C^* = \{ \vec{v}^* \mid \vec{v} \in \C \} \cup \set{\vec{0},\vec{1}}$ and an instance $(\C^*,k+2, 2t)$
of BKMP. Observe that if $(\C, k,t)$ has a solution, then $(\C^*,k+2, 2t)$ has a solution as well, which 
satisfies additional assumptions as above.
Conversely, 
if $(\C^*,k+2, 2t)$ has a solution, then this solution contains 
$\vec{0}$ and $\vec{1}$ because for all $\vec{v}^* \in \C^*$ we have $\lOneNorm{\vec{v}^* - \vec{0}} = 
\lOneNorm{\vec{v}^* - \vec{1}} = m + h > 2t$. It also follows that in the partition $\D_1, \ldots, \D_{k+2}$
two sets that correspond to $\vec{0}$ and $\vec{1}$ respectively are singletons.
Therefore, the solution for $(\C^* \cup \{ \vec{0},\vec{1} \},k+2, 2t)$ can be transformed to a solution of $(\C, k,t)$.

Let $\C = \{ \vec{v}_1, \ldots, \vec{v}_n\} \subseteq \{0,1\}^{2m+2h}$. 
Let $M = 2m+2h$, $N = \frac{M-n}{2}$ ($M = n  + 2N$) and assume that $n$ is even.
We know that $h > t$ and hence $N > 2t$.

We define $\lang_{\C}$ over the alphabet $\Sigma = \{a_1, \ldots, a_{M}\}$ as follows:
for all $a_i, a_j \in \Sigma$, 
$w \in \Sigma^\omega$ 
\begin{itemize}
\item if $i \leq n$, we set $\lang_{\C}(a_ia_j w) = \vec{v}_i[j]$. 
\item if $i \in \{n+1,  n+N\}$, we set $\lang_{\C}(a_ia_j w) = 0$. 
\item if $i \in \{n+N, M\}$, we set $\lang_{\C}(a_ia_j w) = 1$. 
\end{itemize}
This language can be defined with a tree-like $\flimavg$-automaton $\aut_{\C}$. In this automaton states reached over $a_1, \ldots, a_{n}$, i.e.,
the states $\delta(q_0, a_1), \ldots, \delta(q_0, a_n)$ correspond to vectors $\vec{v}_1, \ldots, \vec{v}_n$ via 
$\expected(\aut_{\C} \mid a_i a_j) = \vec{v}_i[j]$.
The vectors $\vec{0}$ and $\vec{1}$ are repeated $N$ times.

First, consider a solution $\vec{u}_1, \ldots, \vec{u}_k$ and $\D_1, \ldots, D_k$ for $(\C, k, t)$, and
$\vec{u}_{1} = \vec{0}$, $\vec{u}_2 = \vec{1}$.
We construct an automaton $\aut_{\D}$ with an initial state $q_0$ and $k$ states $q_1, \ldots, q_k$
First, for all $i,j \in \set{1,\ldots, M}$ we set $\delta(q_0, a_i) = q_j$  if $\vec{v}_j \in \D_i$, i.e., states $q_1, \ldots, q_k$ correspond to vectors $\vec{u}_1, \ldots, \vec{u}_k$. 
Next, $q_1, q_2$ correspond to vectors $\vec{0}$ and $\vec{1}$ respectively and hence 
for all $a \in \Sigma$ we set $\delta(q_1, a_i) =  q_1$ with the weight $0$ and $\delta(q_2, a_i)=q_2$ with the weights $1$.
Finally, for all $i>2$ and all $a_j \in \Sigma$, we define 
	$\delta(q_i, a_j) = q_1$ if $\vec{u}_i[j] = 0$ and 
	$\delta(q_i, a_j) = q_2$ otherwise.
Observe that $\sum_{j=1}^k \sum_{\vec{v} \in \D_j} \lOneNorm{\vec{v} - \vec{u}_j} \leq t$ implies that $\aut_{\D}$ is 
$\frac{t}{nM}$-approximation of $\lang_{\C}$.

Conversely, assume that $\aut'$ is the optimal approximation of $\lang_{\C}$ among automata with $k+1$ states, i.e.,
$\aut'$ $\epsilon$-approximates $\lang_{\C}$, where $\epsilon \leq \frac{t}{nM}$, and no automaton with at most $k+1$ states
approximates $\lang_{\C}$ with smaller $\epsilon$.

Consider vectors $\vec{p}_1, \ldots, \vec{p}_{M} \in \R^{M}$ defined as follows:
for all $i,j \in \set{1,\ldots, M}$ we put 
$\vec{p}_i[j] = \expected(\aut' \mid a_ia_j \Sigma^{\omega})$. 
Since $\aut'$ $\frac{t}{nM}$-approximates $\lang_{\C}$ we know that 
\[ \sum_{i=1}^n \lOneNorm{\vec{p}_i - \vec{v_i}}
+ \sum_{i=n+1}^{n+N} \lOneNorm{\vec{p}_i - \vec{0}} 
+ \sum_{i=n+N+1}^{M} \lOneNorm{\vec{p}_i - \vec{1}} \leq  \frac{t M}{n}
\tag{*}
\label{eq:total-sum}
\] 

We need to show that there are at most $k$ different vectors $\vec{p}_i$ and among these vectors
there are $\vec{0}$ and $\vec{1}$. 

First, observe that all states of $\aut'$ are reachable from the initial state with some word of length at most $2$. Indeed, we can change states, which are in distance $2$ from the initial state, 
to bottom SCCs with the values being the expected value from the old state. 
Second, the initial state cannot be a bottom SCC and hence all bottom SCCs are in distance $1$ or $2$ from the initial state. We can change values of all bottom SCCs to $0$ or $1$.
To see that observe that we can take values of all bottom SCC and consider these values as variables $x_1, \ldots, x_l$. Then, these values appear 
as convex combinations in \eqref{eq:total-sum}. However, since all constants in \eqref{eq:total-sum} are $0$ and $1$, then \eqref{eq:total-sum} is minimal under the substitution mapping $x_1, \ldots, x_l$ to $\set{0,1}$.

Knowing that all bottom SCC have value $0$ and $1$, we observe that we need to have precisely two bottom SCC $s_0, s_1$ of the values $0$ and $1$.
It only improves the approximation if $s_0$ is the successor of the initial state over $a_{n+1}, \ldots, a_{n+N}$ and 
$s_1$ is the successor of the initial state over $a_{n+N+1}, \ldots, a_{n+2N}$.
Therefore, all vectors 
\[
\begin{split}
\vec{p}_{n+1} =  \ldots =  \vec{p}_{n+N} = \vec{0} \\
\vec{p}_{n+N+1} = \ldots = \vec{p}_{M} = \vec{1}
\end{split}
\tag{**}
\label{eq:ZeroOne}
\]

Now, suppose that the initial state $q_0$ of $\aut'$ has a self-loop over $a_i$. 
We know that $i \leq n$ as otherwise we have $\expected(\aut') \in \{0,1\}$ and
$\expected(\lang_{\C}) = \frac{1}{2}$. A contradiction. 
Consider vector $\vec{q}$ defined as $\vec{q}[i] = \expected(\aut' \mid a_i \Sigma^{\omega})$. 
Observe that due to \eqref{eq:ZeroOne} we have
$q[n+1], \ldots, q[n+N] = 0$ and $q[n+N+1], \ldots, q[n+2N] = 1$.
Condition \eqref{cond:alignment} of Modified BKMP implies that $\lOneNorm{\vec{q} - \vec{v}_i} \geq 2h$, while
condition \eqref{cond:equalZero} implies that $\lOneNorm{\vec{1} - \vec{v}_i} = \lOneNorm{\vec{0} - \vec{v}_i} =  \frac{M}{2}$. 
Since $h > n$, we have $2h > \frac{M}{2}$ and hence changing the self-loop $(q_0,a_i,q_0)$ to a transition  $(q_0, a_i, q')$ to the state $q'$ that corresponds to $\vec{1}$ gives us a better approximation of $\lang_{\C}$. 
A contradiction.
It follows that $q_0$ has at most $k$ successors.

Then, $\vec{p}_1, \ldots, \vec{p}_n$ give us a solution to $\C$. First, observe that we can approximate $\vec{p}_1, \ldots, \vec{p}_n$  with vectors from $\{0,1\}^{M}$ without increasing 
the left hand side of \eqref{eq:total-sum}. Hence, we assume that they belong to $\{0,1\}^{M}$. Second,
we select different vectors among  $\vec{p}_1, \ldots, \vec{p}_n$  as $\vec{u}_1, \ldots, \vec{u}_k$ and
define $\D_i = \{ \vec{v}_j \mid \vec{p}_j = \vec{u}_i \}$. The,  \eqref{eq:total-sum} implies that this is a solution to $(\C,k,t)$.
\end{proof}

\section{NP{}-completeness of rigid approximate minimization}
\MinimalRigidNPComp*
\begin{proof}[Proof sketch]
The problem is in $\NP$. 
We can non-deterministically pick an automaton $\aut'$ with $n$ states and check whether it is a rigid $\epsilon$-approximation of ${\aut}$.
Observe that having two automata $\aut_1,\aut_2$, we can check whether $\aut_1$ is a rigid $\epsilon$-approximation of ${\aut}_2$ in polynomial time.
Indeed, let $P$ be a set of pairs of states defined as follows: $(q,s) \in P$ if and only if $q$ is a state of $\aut_1$, $s$ is a state of $\aut_2$
and they are both reached in the respective automata from the (respective) initial states over a common word $u$. 
Observe that
$\aut_1$ is a rigid $\epsilon$-approximation of ${\aut}_2$ if and only if for every pair $(q,s)$
the automaton $\aut_1$ starting from $q$ $\epsilon$-approximates $\aut_2$ starting from the state $s$. 
We have shown in the proof of Theorem~\ref{th:approximateNPComp} (in Appendix~\ref{app:appendixA}) that we can decide $\epsilon$-approximation in polynomial time.
Therefore, we can decide rigid $\epsilon$-approximation in polynomial time as well.
\medskip 

We show that the problem is $\NP$-hard.
To ease the presentation, we define the following $\frac{1}{4}$-vector cover problem, which is an intermediate step in our reduction. 
\medskip

\noindent\emph{The $\frac{1}{4}$-vector cover problem}: 
Given $n,k \in \N$ and vectors $\C = \set{\vec{v}_1, \ldots, \vec{v}_n } \subseteq \{0, \frac{1}{2}, 1\}^n$, decide
whether there exist $\vec{u}_1, \ldots, \vec{u}_k \in \R^{n}$ such that for every vector $\vec{v} \in \C$ there is $j$ such that 
$\lInfNorm{\vec{v} - \vec{u}_j} \leq \frac{1}{4}$.
\medskip

The $\frac{1}{4}$-vector cover problem is related to BKMP presented in the proof of Theorem~\ref{th:approximateNPComp}.
We show two reductions, which together show $\NP$-hardness.
\medskip

\noindent \emph{The $\frac{1}{4}$-vector cover problem reduces to the rigid approximate minimization problem}.
Let  $\C = \set{ \vec{v}_1, \ldots, \vec{v}_n } \subseteq \{0, \frac{1}{2}, 1\}^n$.
We define $\lang_{\C}$ over the alphabet $\Sigma = \{a_1, \ldots, a_{m}\}$ such that for all $a_i, a_j \in \Sigma$, 
$w \in \Sigma^\omega$ we have $\lang_{\C}(a_ia_j w) = \vec{v}_i[j]$. Such $\lang_{\C}$ can be defined with a 
a tree-like $\flimavg$-automaton $\aut_{\C}$ defined as follows.
It has three single-state bottom SCCs: $p_0, p_{0.5}, p_1$ with the expected values $0, 0.5$ and $1$ respectively.
The automaton $\aut_{\C}$ moves over letters $a_i$ to different states. 
Then, over any two-letter word $a_i a_j$ the automaton $\aut_{\C}$ ends up in a single-state bottom SCC of the value $\vec{v}_i[j]$. 
Therefore, this automaton has $n+4$ states: the initial state $q_0$, $n$ different successors $s_1, \ldots, s_n$ of $q_0$ and the states $p_0, p_{0.5}, p_1$.

Let $\epsilon = \frac{1}{4}$. We show that an instance $n,k,\C$ 
of the $\frac{1}{4}$-vector cover problem has a solution  if and only if $\aut_{\C}$ has a rigid $\epsilon$-approximation with $k+3$ states.

First, assume that there is $\aut'$ has at most $k+3$ states and 
for all words $u \in \Sigma^*$ we have $\expected(|\valueL{\aut_{\C}} - \valueL{\aut'}| \mid u\Sigma^{\omega}) \leq \epsilon$.
Observe that the initial state of $\aut'$ has at most $k$ different successors.
To see that consider functions $\lang^{u}(w) = \valueL{\aut_{\C}}(u w)$ defined for $u \in \Sigma^*$. 
If for some $u_1, u_2, u \in \Sigma^*$ we have $\expected(|\lang^{u_1 u} - \lang^{u_2 u}|) > 2 \epsilon$, then words $u_1$ and $u_2$ lead (from the initial state) to different states in $\aut'$.
Using this observation, we can show that $\aut'$ has at least $3$ states that cannot be successors of the initial state: the initial state and
(at least two) states that corresponding to bottom SCCs.

Finally, we define vectors $\vec{u}_1, \ldots, \vec{u}_k$ from the successors of the initial state of $\aut'$. Formally,
for $i \in \{1, \ldots, n\}$ we define a vector $\vec{y}_i$ as $\vec{y}_i[j] = \expected(\valueL{\aut'} \mid a_ia_j \Sigma^{\omega})$. 
Note that there are at most $k$ different vectors among $\vec{y}_1, \ldots, \vec{y}_n$ and we take these distinct vectors as $\vec{u}_1, \ldots, \vec{u}_k$.
The condition $\expected(|\valueL{\aut_{\C}} - \valueL{\aut'}| \mid u\Sigma^{\omega}) \leq \epsilon$ implies that 
$\vec{u}_1, \ldots, \vec{u}_k$ satisfy the $\frac{1}{4}$-vector cover problem.

Conversely, assume that the instance $n,k,\C$ has a solution. Observe that w.l.o.g we can assume that
the solution vectors $\vec{u}_1, \ldots, \vec{u}_k$ belong to $\set{0.25,0.75}^n$.
Based on this solution we define an automaton with $k+3$ states such that 
the successors of the initial state correspond to vectors $\vec{u}_1, \ldots, \vec{u}_k$ and there are two bottom SCC of the values $0.25$ and $0.75$.

Formally, we define $\aut'$ as follows.
Let $f \colon \set{1,\ldots, n} \to \set{1,\ldots, k}$ be the mapping of vectors $\vec{v}_i \in \C$ to $\vec{u}_j$ such that $\lInfNorm{\vec{v}_i - \vec{u}_j} \leq \frac{1}{4}$  (and then $f(i) = j$).
We define $\aut'$ with states $q_0, q_1, \ldots, q_k, s_1, s_2$ such that 
$q_0$ is the initial state, 
$q_1, \ldots, q_k$ are the successors of the initial state, and 
$s_1, s_2$ are single-state SCCs of the values $0.25$ and $0.75$ respectively.
We define the transition function as follows.
For all $a_i \in \Sigma$ we set $\delta(q_0, a_i) = q_{f(i)}$.
Then for every $i,j \in \set{1, \ldots, n}$ we set $\delta(q_i, a_j) = s_1$ if $\vec{u}_i[j] = 0.25$ and 
$\delta(q_i, a_j) = s_2$ otherwise ($\vec{u}_i[j] = 0.25$).
Now, the fact that  $\vec{u}_1, \ldots, \vec{u}_k$ is the solution to the $\frac{1}{4}$-vector cover problem implies that
$\aut'$ is a rigid $\frac{1}{4}$-approximation of ${\aut_{\C}}$.
\medskip

\noindent \emph{The dominating set problem reduces to the $\frac{1}{4}$-vector cover problem}.
Consider a graph $G = (V,E)$ with $V = \{ b_1, \ldots, b_n \}$ and $k \in \N$. 
We assume that $G$ has no cycles of length less than $5$; this restriction does not influence $\NP$-hardness of the problem, since each edge can be substituted with a path of odd length greater than $5$.
Denote by $\neibr_k(b)$ the set of all nodes of $G$ connected to $b$ with a path of length at most $k$;  $b$ is connected with itself with a path of length $0$ and hence 
$b \in \neibr_k(b)$.
We define vectors $\vec{v}_1, \ldots, \vec{v}_n \in \{0, \frac{1}{2}, 1\}^n$ as follows.
For $i,j \in \{1,\ldots, n\}$ we set
\begin{itemize}
\item  $\vec{v}_j[i] = 1$ if $i=j$,
\item  $\vec{v}_j[i] = 0.5$, if $i \neq j$, but $b_i \in \neibr_2(b_j)$, and 
\item $\vec{v}_j[i] = 0$ otherwise.
\end{itemize}
We claim that there exist  $\vec{u}_1, \ldots, \vec{u}_k \in \R^{n}$ as in the problem statement if and only if $G$ has a dominating set of the size $k$.
\smallskip

Assume that $G$ has a dominating set $d_1, \ldots d_k$. 
Consider $\vec{u}_1, \ldots, \vec{u}_k$ such that for all $i \in \set{1,\ldots,k}$ we have
\begin{itemize}
\item $\vec{u}_i[j] = \frac{3}{4}$ if $i = j$ or  $(d_i,b_j) \in E$, and
\item $\vec{u}_i[j] = \frac{1}{4}$ otherwise. 
\end{itemize}
Observe that for all 
$i \in \set{1,\ldots,k}$ and $j \in \set{1,\ldots,n}$, if $i$ satisfies  $i = j$ or $(b_j,d_i) \in E$, then $\lInfNorm{\vec{v}_j - \vec{u}_i} \leq \frac{1}{4}$. 
Therefore, the vectors $\vec{u}_1, \ldots, \vec{u}_k$ solve the $\frac{1}{4}$-vector cover problem.

Conversely, assume that there exist $\vec{u}_1, \ldots, \vec{u}_k \in \R^{n}$ that solve the $\frac{1}{4}$-vector cover problem. 
Let $\vec{u}_j$ be a vector such that for $\vec{v}_{m[1]}, \ldots, \vec{v}_{m[l]}$ we have $\lInfNorm{\vec{u}_j - \vec{v}_{m[i]}} \leq \frac{1}{4}$.
Note that we can assume that all coefficients of $\vec{u}_j$ are $\frac{1}{4}$ and $\frac{3}{4}$.
We claim that for some $i$ nodes $b_{m[1]}, \ldots, b_{m[l]}$ of $G$ belong to $\neibr_1(b_{m[i]})$ for some $i$. 
To see that observe that  the distance between any two nodes among $b_{m[1]}, \ldots, b_{m[l]}$ is at most $2$.
Indeed, for all $k$ the component $m[k]$ of $\vec{v}_{m[k]}$ is $1$ and hence this component of $\vec{u}_j$ is $\frac{3}{4}$.
That in turn implies that the component $m[k]$ of $\vec{v}_{m[1]}, \ldots, \vec{v}_{m[l]}$  is $1$ or $\frac{1}{2}$ and hence 
$b_{m[k]}$ belongs to $\neibr_2(b_{m[1]}), \ldots, \neibr_2(b_{m[l]})$.
Since there are no short cycles in $G$ and all distances are bounded by $2$, there has to $i$ such that 
$b_{m[1]}, \ldots, b_{m[l]}$ of $G$ belong to $\neibr_1(b_{m[i]})$. 
Then, we define $d_j$ as $b_i$.

Note that $d_j$ dominates all nodes $b_{m[1]}, \ldots, b_{m[k]}$, which correspond to vectors  $\vec{v}_{m[1]}, \ldots, \vec{v}_{m[l]}$.
Therefore, nodes $d_1, \ldots, d_k$ picked as above form a dominating set in $G$.
\end{proof}

\section{Estimating minimal sample size in passive learning}\label{app:samplesize}

The probability that a single words of length $l$ is not generated by a random sample $S$ with $s$ examples, generated w.r.t. a distribution $\uniformFin{\lambda}$, can be bounded by
\[
\left(1-\frac{1}{|\Sigma|^{l}(1-\lambda)^l\lambda}\right)^s
\]

We want to compute a sample size $s$ such that the probability that there is a word of size $l$ not in this sample is at most $\epsilon$. This can be, very roughly, represented by the following inequality:
\[
|\Sigma|^l \left(1-\frac{1}{|\Sigma|^{l}(1-\lambda)^l\lambda}\right)^s < \epsilon
\]
which we can conveniently rewrite as 
\[
\left(1-\frac{1}{|\Sigma|^{l}(1-\lambda)^l\lambda}\right)^{|\Sigma|^{l}(1-\lambda)^l\lambda \cdot s/(|\Sigma|^{l}(1-\lambda)^l\lambda)} < 
\frac{\epsilon}{|\Sigma|^l}
\]

By the fact that $(1-\frac{1}{x})^x < e^{-1} $, the above equality is a consequence of the following one 
\[
e^{-s/(|\Sigma|^{l}(1-\lambda)^l\lambda)} < \frac{\epsilon}{|\Sigma|^l}
\]
which is equivalent to 
\[
e^{s/(|\Sigma|^{l}(1-\lambda)^l\lambda)} > \frac{|\Sigma|^l}{\epsilon}
\]
Now we apply the natural logarithm.
\[
{s/(|\Sigma|^{l}(1-\lambda)^l\lambda)} > \ln \frac{|\Sigma|^l}{\epsilon}
\]
so
\[
s > |\Sigma|^{l}(1-\lambda)^l\lambda  \cdot \ln \frac{|\Sigma|^l}{\epsilon}
\]

For $(E, n)$-samples, the estimation is similar.

\end{document}